\documentclass[a4paper,10pt]{article}
\usepackage{graphicx} 
\usepackage[utf8]{inputenc}
\usepackage[T1]{fontenc}
\usepackage{comment}
\usepackage{float}
\usepackage{natbib}
\usepackage{pdflscape}  
\usepackage{xcolor}
\usepackage{booktabs} 

\usepackage{a4wide}
\usepackage{indentfirst}
\usepackage{authblk}

\usepackage{amsmath,amssymb}
\usepackage{amsthm}
\usepackage[colorlinks=true, linkcolor=black, citecolor=black, urlcolor=blue]{hyperref}
\usepackage{footmisc} 
\usepackage{subfig}
\usepackage{graphicx}
\usepackage{subcaption}
\usepackage{chngcntr}
\counterwithin{table}{section}
\counterwithin{figure}{section}

\usepackage[margin=1.2in]{geometry}
\usepackage{graphicx}
\usepackage{listings}
\usepackage[utf8]{inputenc}
\usepackage{appendix}
\usepackage{shuffle}
\numberwithin{equation}{section}
\newtheorem{teo}{Theorem}[section]
\newtheorem*{teo*}{Theorem}
\newtheorem*{prop*}{Proposition}
\newtheorem*{corol*}{Corollary}
\newtheorem{prop}[teo]{Proposition}

\newtheorem{lema}[teo]{Lemma}
\newtheorem{defi}[teo]{Definition}

\theoremstyle{definition}

\newtheorem{ex}[teo]{Example}

\newtheorem{remark}[teo]{Remark}

\newcommand{\BS}{\operatorname{BS}}

\newcommand{\e}{\mathrm{e}}

\newcommand{\floor}[1]{\left\lfloor #1 \right\rfloor}

\newcommand{\bfa}{\mathbf{a}}
\newcommand{\bfb}{\mathbf{b}}
\newcommand{\E}{\mathbb{E}}
\newcommand{\F}{\mathcal{F}}

\newcommand{\N}{\mathbb{N}}
\renewcommand{\P}{\mathbb{P}}

\newcommand{\R}{\mathbb{R}}
\newcommand{\X}{\mathbb{X}}
\newcommand{\bX}{\mathbf{X}}
\newcommand{\bY}{\mathbf{Y}}
\newcommand{\bB}{\mathbf{B}}

\newcommand{\1}{\mathbf{1}}

\newcommand{\eps}{\varepsilon}

\newcommand{\blfootnote}[2]{%
  \begingroup
  \renewcommand\thefootnote{#1}%
  \addtocounter{footnote}{1}%
  \footnotetext{#2}%
  \addtocounter{footnote}{-1}%
  \endgroup
}

\title{Volatility Modeling in Markovian and Rough Regimes: Signature Methods and Analytical Expansions}

\author{Elisa Alòs\textsuperscript{*}, Òscar Burés\textsuperscript{†\,§}, Rafael de Santiago\textsuperscript{§} and Josep Vives\textsuperscript{†\,‡}}
\date{\today}

\begin{document}

\maketitle

\blfootnote{*}{Department of Economics and Business, Universitat Pompeu Fabra and Barcelona School of Economics. Ramón Trias Fargas 25-27, 08005, Barcelona, Spain.}
\blfootnote{†}{Departament de Matemàtica Econòmica, Financera i Actuarial, Universitat de Barcelona. Diagonal 690--696, 08034 Barcelona, Spain.}
\blfootnote{§}{Department of Managerial Decision Sciences, IESE Business School. Av. Pearson 21, 08034 Barcelona, Spain.}
\blfootnote{‡}{Institut de Matemàtiques, Universitat de Barcelona. Gran Via de les Corts, 585, 08007 Barcelona, Spain.}

\begingroup
\renewcommand{\thefootnote}{}  
\renewcommand{\footnotemargin}{0pt}  
\footnotetext{%
 \noindent\hspace{0pt}Òscar Burés supported by program AGAUR-FI ajuts (2025 FI-1 00580) from the Department of Research and Universities of the Government of Catalonia and the co-funding of the European Social Fund Plus (ESF+).}
\endgroup

\begin{abstract}
We study two complementary methodologies for calibrating implied volatility surfaces: analytical approximations and data-driven models based on rough path theory. On the analytical side, we revisit a second-order asymptotic expansion for the Heston model, and we propose a new, VIX-based calibration scheme for the rough Bergomi model. Both methods yield highly accurate and computationally efficient calibration formulas when the underlying dynamics are well specified. In parallel, we develop a signature-based approach in which volatility is represented as a linear functional of the truncated signature of a primary stochastic process, providing a flexible and model-agnostic alternative.

Our numerical experiments compare the two approaches across both Markovian and non-Markovian settings. In the Heston case, signature-based models achieve a level of accuracy comparable to analytical expansions. In the rough Bergomi setting, using a fractional Brownian motion as the primary process, the signature approach continues to perform strongly and in some cases improves upon the Markovian specification, reflecting its ability to capture more complex temporal dependencies. Overall, the results illustrate that analytical methods are highly effective when the model is correctly specified, while signature-based methods offer a robust and flexible framework for calibration across a wider range of volatility dynamics.
\end{abstract}




\section{Introduction}

The realization that constant-volatility models could not account for effects like clustering, heavy tails, and \emph{smiles} prompted the extension of the Black–Scholes framework to include stochastic volatility. Early models addressed these limitations by introducing a second (correlated) Brownian motion to govern the volatility, yielding more realistic asset price dynamics. This line of research, initiated by \cite{hullwhite}, \cite{wiggins87}, \cite{steinstein}, and \cite{heston93}, has given rise to a broad and influential literature in mathematical finance.

A particularly influential part of the literature has focused on developing closed-form approximation formulas for implied volatility through asymptotic expansions and perturbation methods applied to parametric models, such as Heston and SABR. These approximations aim to facilitate the calibration of model parameters to market-observed implied volatility surfaces. Efficient calibration is a central task in financial practice, as option prices are typically quoted via implied volatilities, and model parameters must be inferred by inversion. Contributions to this area include, among others, \cite{HKLW02}, \cite{FPSS2003}, \cite{MS07}, \cite{DSFS2008}, \cite{AS09}, \cite{BGM09, BGM10a, BGM10b}, \cite{FJM10}, \cite{FJ11}, \cite{FJL11}, \cite{alos12}, \cite{LPP13} and \cite{ADSV15}.

While these parametric models and their expansions have proven theoretically elegant and computationally tractable, they also come with important limitations. Relying on a fixed model structure may restrict flexibility and hinder the model’s ability to capture certain stylized features observed in market data. In addition, the presence of multiple stochastic drivers often increases computational complexity, requiring the use of numerical tools such as Fourier transforms, Monte Carlo methods, or finite-difference PDE solvers. Moreover, a growing body of empirical evidence suggests that volatility exhibits rough, fractional-like behavior that traditional Markovian models struggle to reproduce. This has led to increased interest in fractional (rough) volatility models, such as those introduced in \cite{CR98}, \cite{ALV}, and \cite{fukasawa17}, which more accurately reflect the observed structure of the implied volatility surface and align with empirical evidence on the roughness of volatility paths (e.g., \cite{bayerfrizgatheral16}, \cite{GJR18}).

In response to these challenges, a more recent line of research has embraced non-parametric, data-driven approaches that aim to learn volatility dynamics directly from observed paths, rather than specifying a rigid structural model. A prominent example of this paradigm is the use of rough path theory and signature methods, originally introduced in \cite{chen58} and rigorously developed in \cite{lyons98}. The signature of a path (to be defined precisely below) consists of its iterated integrals and encodes its temporal features into a rich algebraic structure. Over the years, these ideas have found applications in machine learning, time series analysis, and mathematical finance (e.g., \cite{buehler20}, \cite{perezarribas20}, \cite{cuchiero23}). Because it relies on functional features rather than a parametric form, signature-based modeling is well suited for capturing behaviors observed in rough volatility models.

In this paper, we consider two complementary approaches to the calibration of implied volatility surfaces. On the one hand, we present analytical approximations for both the Heston and rough Bergomi models. For Heston, we revisit a second-order asymptotic expansion, while for rough Bergomi we introduce a new calibration scheme based on VIX-implied information. These methods yield highly accurate formulas when the underlying model is well specified. On the other hand, we study signature-based approximations in which the volatility is represented as a linear functional of the truncated signature of a primary stochastic process (taken to be, respectively, a Heston variance process and a fractional Brownian motion). This approach does not rely on a fixed parametric specification and provides a flexible framework capable of adapting to different volatility dynamics. Rather than contrasting the two methodologies, our aim is to highlight their respective strengths and to assess their performance across both Markovian (Heston) and non-Markovian (rough Bergomi) regimes.

To make the paper self-contained, Section~\ref{2orderApprox} presents the analytical approximations for both models: we first review the second-order expansion technique for the Heston model introduced in \cite{ADSV15}, highlighting its mathematical structure and practical strengths, and then develop a new calibration scheme for the rough Bergomi model, based on VIX-implied volatility information. Section~\ref{pathsignatures} introduces the core elements of rough path theory needed to define and compute path signatures. In Section~\ref{SigBasedModel}, we develop the signature-based stochastic volatility model, describe its numerical implementation, and outline the training procedure. Section~\ref{numerical_Heston} reports numerical results for the Heston setting, including both uncorrelated and correlated dynamics as in~\cite{ADSV15}, and compares them with the corresponding analytical approximation. Section~\ref{numerical_rBergomi} presents the results for the rough Bergomi model, using a fractional Brownian motion as the primary process, and provides the analogous comparison.

\section{Analytic Calibration Methods} \label{2orderApprox}

In this section we explore two alternative ways of calibrating analytically the implied volatility surface generated by the Heston model and the rough Bergomi model. For the Heston model, we follow the approach in \cite{ADSV15}, which calibrates the model parameters by solving a system of non-linear equations. For the rough Bergomi model, we introduce a new way to calibrate the implied volatility surface using short-maturity information. 

\subsection{A Second-Order Approximation to the Heston Implied Volatility}\label{ASV_IVolApproxim}

We provide here a concise overview of the second-order approximation derived in \cite{ADSV15}. For full proofs and theoretical justifications, we refer the reader to the original paper. Our purpose is to use the estimates obtained with this methodology as a benchmark for comparison with the signature-based models that are introduced later. 

Assume that, for $t \in [0,T]$, the stock price follows the dynamics
\begin{equation} \label{stockpriceSDE}
dS_t = rS_t dt + \sigma_t S_t d(\rho W_t + \sqrt{1-\rho^2} B_t) 
\end{equation}
under a risk neutral probability $\P$, where $r \geq 0$ is the constant instantaneous interest rate, $W$ and $B$ are independent standard Brownian motions defined in the complete probability space $(\Omega, \F, \P)$, and $\rho \in (-1,1)$. We also assume that the volatility process $\sigma_t$ satisfies
\begin{equation} \label{Hestonvolatility}
d\sigma_t^2 = \kappa(\theta -  \sigma_t^2) dt +\nu \sqrt{\sigma_t^2 } dW_t,
\end{equation}
with $2\kappa \theta \geq \nu^2$. We denote by $\F^W = \{\F_t^W; t \in [0,T]\}$ and $\F^B = \{\F_t^B; t \in [0,T]\}$ the filtrations generated, respectively, by $W$ and $B$, and we define $\F$ as the collection of sigma algebras $\F^W_t \vee \F^B_t$ for each $t\in [0,T]$, that is, $\F := \F^W \vee \F^B$. Equations (\ref{stockpriceSDE}) and (\ref{Hestonvolatility}) constitute what is known as the Heston model. 

If we let $X_{t}:=\ln S_t$, the price of a European call option at time $t$ with strike $K$ and maturity $T$ is given by
\[
V_t=e^{-r(T-t)}E_t[(e^{X_T}-K)^+],
\]
where $E_t[\cdot] := E[\cdot|\F_t]$. For a constant volatility $\sigma$, and letting $k=\ln K$, the above general expression has the well-known analytical solution  
\[
V_t = \BS(T,t,X_t,k,\sigma) = e^{X_t} \Phi(d_+) - e^{k-r(T-t)}\Phi(d_{-}), 
\]
where $\Phi$ is the cumulative distribution function of a standard Gaussian random variable, and
\[
d_{\pm} = \frac{X_t-k+r(T-t)}{\sigma\sqrt{T-t}} \pm \frac{\sigma}{2}\sqrt{T-t}.
\]

When volatility is stochastic, the Black-Scholes formula no longer provides an exact solution. However, we can define an implied volatility approximation by evaluating the Black-Scholes formula at
\begin{equation} \label{nonconstant}
V_t \approx \BS(T,t,X_t,k,v_t),
\end{equation}
where  
\[
v_t = \sqrt{\frac{1}{T-t}\int_t^T E_t[\sigma_s^2] ds}
\]
represents the square root of the expected average variance over the remaining life of the option. 

Let $V^{mkt}_t$ be the market price at time $t$ of a European call option with maturity $T$ and strike $K$. As the $\BS$ function is invertible in the argument $v_t$, we can define the implied volatility as the unique $I(T,K)$ satisfying the equality
\[
\BS(T,t,X_t, k, I(T,K)) = V^{mkt}_t.
\]

Since even in the simplest stochastic volatility models it is generally not feasible to derive closed-form expressions for the implied volatility surface, a variety of approximation methods have been developed to estimate $I(T,K)$. 

The method in \cite{ADSV15} proceeds as follows. First, an approximation to the price of a European call is derived. Based on this, the following second-order expansion of the implied volatility is obtained:
\[
I(K,T) \approx v_0 + I_1(K,T) + I_2(K,T)  
\] 
where
\begin{align*}
I_1(K,T) & = \frac{\rho}{2v_0 T}\left(1- \frac{d_+(K,T)}{v_0 \sqrt{T}} \right)\E\left[\int_0^T \sigma_s d \langle M,W\rangle_s \right] \\
I_2(K,T) & = \frac{1}{8v_0 T} \left( \frac{d_+(K,T)^2}{v_0^2 T} - \frac{d_+(K,T)}{v_0\sqrt{T}} - \frac{1}{v_0^2 T} \right)\E\left[ \int_0^T d \langle M, M\rangle_s\right]  
\end{align*}
and $M_t = \int_0^T E_t(\sigma_s^2) ds$.
Finally, using the above expressions together with the term structure of at-the-money (ATM) options, the following asymptotic results are derived:
\begin{itemize}
\item When the call is close to maturity ($T \to 0$), the implied volatility can be approximated as:
\begin{equation}\label{IV-at-zero}
I(0, K) \approx \sigma_0 \,-\, \frac{\rho\nu }{4\sigma _{0}}\left( x-k\right)
\,+\, \frac{\nu^{2}}{24\sigma _{0}^{3}}\ (x-k)^{2}.
\end{equation}

\item When the European call is far away from maturity ($T \to \infty$):
\begin{equation}\label{IV2-at-infty}
I(T, K) \approx \sqrt{\theta }\left( 1+\frac{\nu \rho }{4\kappa }-\frac{%
\nu ^{2}}{32\kappa ^{2}}\right) +\left( \frac{\sigma _{0}^{2}-\theta }{%
2\kappa \sqrt{\theta }}\,+\,\nu \rho \frac{\sigma _{0}^{2}-2\theta }{4\kappa
^{2}\sqrt{\theta }}\,-\,\nu ^{2}\frac{\sigma _{0}^{2}-\frac{5}{2}\theta
+4\kappa }{32\sqrt{\theta }\kappa ^{3}}\right) \frac{1}{T}.
\end{equation}

\item When the option is at-the-money ($x = k - rT$), the expansion around $\sigma_0$ is:
\begin{equation}
I(T, K) \approx \sigma _{0}+\frac{3\sigma _{0}^{2}\rho \nu -6\kappa
(\sigma _{0}^{2}-\theta )-\nu ^{2}}{24\sigma _{0}}\,T.  \label{Taylor-at-0}
\end{equation}
\end{itemize}

\noindent These approximations are then used as follows. By fitting a linear equation to the ATM implied volatilities for different values of $T$, the values of $\sigma_0$ and  
$$ 
\frac{3\sigma _{0}^{2}\rho \nu -6\kappa (\sigma _{0}^{2}-\theta )-\nu ^{2}}{24\sigma _{0}}
$$
are obtained from (\ref{Taylor-at-0}).

By fitting a linear equation to the implied volatilities for short maturities as a function of the log-moneyness $(x-k)$, equation (\ref{IV-at-zero}) provides the value of $\rho\nu$. Notice that for strikes close to the money, the linear component is stronger than the curvature, making it difficult to estimate $\nu$ directly.

For large values of $T$, an equation is fitted to the implied volatilities as a function of $1/T$, and the values of 
$$
\sqrt{\theta }\left( 1+\frac{\nu \rho }{4\kappa }-\frac{\nu^{2}}{32\kappa ^{2}}\right) 
$$
and
$$
\left( \frac{\sigma _{0}^{2}-\theta }{%
2\kappa \sqrt{\theta }}\,+\,\nu \rho \frac{\sigma _{0}^{2}-2\theta }{4\kappa
^{2}\sqrt{\theta }}\,-\,\nu ^{2}\frac{\sigma _{0}^{2}-\frac{5}{2}\theta
+4\kappa }{32\sqrt{\theta }\kappa ^{3}}\right) .
$$
are obtained from (\ref{IV2-at-infty}). 

The previous steps provide a system of three equations. By solving them, the full set of Heston parameters ($\sigma_0, \nu, \kappa, \theta, \rho$) is calibrated.     

In particular, with \emph{market} parameters $S_0=100$, $\sigma_0 = 0.2$, $\nu = 0.3$, $\kappa = 3$, $\theta = 0.09$ and $\rho = 0$, the parameters calibrated with the above algorithm are:

\begin{table}[H] 
    \centering
    \begin{tabular}{ccc}
        \hline
        Parameter & True Value & Calibrated Value \\
        \hline
        $\sigma_0$ & 0.2 & 0.200013 \\
        $\nu$ & 0.3 & 0.307340 \\
        $\kappa$ & 3 & 2.998598 \\
        $\theta$ & 0.09 & 0.089960 \\
        $\rho$ & 0 & 0.000000 \\
        \hline
    \end{tabular}
    \caption{Calibrated parameters from \cite{ADSV15} with $\rho = 0$.}
    \label{tab:ASV-calibration}
\end{table}

In the correlated case, with $\rho = -0.5$, the calibrated parameters are: 
\begin{table}[H]
    \centering
    \begin{tabular}{ccc}
        \hline
        Parameter & True Value & Calibrated Value \\
        \hline
        $\sigma_0$ & 0.2 & 0.200016 \\
        $\nu$ & 0.3 & 0.290138 \\
        $\kappa$ & 3 & 2.973728 \\
        $\theta$ & 0.09 & 0.090022 \\
        $\rho$ & $-0.5$ & $-0.504084$ \\
        \hline
    \end{tabular}
    \caption{Calibrated parameters from \cite{ADSV15} with $\rho = -0.5$.}
    \label{tab:ASV-correlated}
\end{table}

These accurate calibrations will serve as a benchmark in Section~\ref{numerical_Heston} for comparison with the signature-based model. A limitation of this method lies in the assumption of taking the Heston model as the underlying dynamics for the volatility process. As a result, its flexibility is constrained by the limitations of that specific model, which may not adequately capture roughness or other complex behaviors observed in real markets.

\subsection{Closed-form calibration of the Rough Bergomi Model} \label{short_time_calib_rb}

Consider the rough Bergomi model 
\begin{equation}\label{rBergomi_equation}
\begin{aligned}
    dS_t &= r S_t\, dt + \sigma_t S_t\, d\left( \rho W_t + \sqrt{1 - \rho^2}\, B_t \right) \\
    \sigma_t^2 &= \sigma_0^2 \exp\left( \eta W_t^H - \frac{1}{2} \eta^2 t^{2H} \right),
\end{aligned}
\end{equation}
where $\eta > 0$, $H \in (0,1)$ and $W^H_t$ is a Volterra-type fractional Brownian motion:
\[
W_t^H := \int_0^t K_H(t,s)\, dW_s,
\]
with $K_H(t,s) = \sqrt{2H}\, (t - s)^{H - \frac{1}{2}}$ for $0<s<t$. Assume that the risk-free rate $r$ is known (it is observable from the market) and that $\rho \neq 0$. 

In this section we propose a new closed-form calibration procedure for the rough Bergomi model that combines short-time asymptotic results for the implied volatility surface with information extracted from VIX-implied volatility, leading to an efficient and easily implementable estimation of the model parameters. In the same way as with the method for the Heston calibration, we assume that the whole implied volatility is available. By this, we mean that $I(T,K)$ can be determined for every maturity $T$ and strike $K$. 

Let $H \in (0,1/2)$. To estimate the set $(\sigma_0, H, \eta, \rho)$, we proceed as follows.
\begin{itemize}
    \item \textbf{Step 1 - Estimation of $H$.} We take two reference strikes on either side of ATM, $K_T^+$ and $K_T^-$, respectively defined by $d_+=0$ and $d_-=0$. The difference $I(T,K_T^+) - I(T,K_T^-)$ captures the skew of the implied volatility smile. For maturities $T_1$ and $T_2$, Alòs et al.\,(2025) show that the Hurst index can be estimated as:
    \[
    \hat{H} = -\frac{1}{2} + \frac{\ln\left( \frac{I(T_1, K^+_{T_1}) - I(T_1, K^{-}_{T_1})}{I(T_2, K^+_{T_2}) - I(T_2, K^-_{T_2})} \cdot \frac{I^2(T_2, K^*)}{I^2(T_1, K^*)}\right)}{\ln\left( \frac{T_1}{T_2}\right)},
    \]
    where $K^*$ denotes the ATM log-strike. As this approximation does not depend on specific model parameters, it can be used without first calibrating a specific volatility model. With rough fractional volatility, $\hat{H}$ provides a quick way to estimate the Hurst parameter from the implied volatility surface.
    
    \item \textbf{Step 2 - Estimation of $\eta$.} To estimate $\eta$, we use the short-time behavior of the implied volatility of VIX options. Let $\Delta$ be 30 trading days. The VIX index at time $T$ is defined as
    \[
    VIX_T = \sqrt{\frac{1}{\Delta}E_T \left[ \int_T^{T+\Delta} \sigma_s^2 ds\right]}.
    \]
    Consider a European option with payoff $(VIX_T - K)_+$. Let the strike be $K=VIX_0$ and denote by $I^{VIX}_T(0)$ the implied volatility of such an option. That is, $I^{VIX}_T(0)$ is the ATM implied volatility (ATMI) of a European option on the VIX index with maturity $T$. Theorem~$8$ in \cite{AlosGLMugu22} proves a property of the short-time behavior of the ATMI in a general setting. Applying this result to the particular case of $I^{VIX}_T(0)$ under rough Bergomi dynamics, the following result is derived in Example~$10.2.3$ of \cite{AlosLorite25}:
    \[
    \lim_{T \to 0}I^{VIX}_T(0) = \frac{\eta \sqrt{2H} \Delta^{H-1/2}}{2(H+\frac{1}{2})}.
    \]
    Once we know $\hat{H}$, the estimator $\hat{\eta}$ is thus computed as
    \[
    \hat{\eta} = \frac{I^{VIX}_T(0)(2\hat{H}+1)}{\sqrt{2\hat{H}}\Delta^{\hat{H}-1/2}}, \quad \text{provided } T \approx 0.
    \]
    
    \item \textbf{Step 3 - Estimation of $\rho$.} The main result in \cite{AlosNuaPrav24} deals with the short-time \emph{skew} of the ATMI (Theorem~$1$). When applied to the rough Bergomi model (Section~$5.2$), the authors show that
    \[
    \lim_{T \to 0} T^{1/2 - H}\partial_K I(0,K^\star) = \frac{2\eta \rho \sqrt{2H}}{3+4H(2+H)}, 
    \]
    where $K^\star$ is the ATM strike at time $0$. Given the estimates $\hat{H}$ and $\hat{\eta}$, and using a finite difference scheme to compute $\partial_k I(T, K^\star)$, we obtain the estimator $\hat{\rho}$ as:
    \[
    \hat{\rho} = \frac{T^{1/2- \hat{H}} \partial_k I(0,K^\star)(3+4\hat{H}(2+\hat{H}))}{2\hat{\eta}\sqrt{2\hat{H}}}, \quad \text{provided } T\approx 0.
    \]
    
    \item \textbf{Step 4 - Estimation of $\sigma_0$.} For this step, only $\hat{H}$ is required. Theorem~$6.5.5$ in \cite{AlosLorite25} shows that the ATM implied volatility of European options corresponding to a rough Bergomi model satisfies the following asymptotic relationship:
    \[
    I(T, K^\star) \approx \sigma_0 + c_1 T^{2H} + O(T^{H+1/2}).
    \]
    Thus, considering several ATM options with different maturities, we can compute $I(T,K^\star)$ and do a regression to obtain $\sigma_0$. This regression also provides $c_1$, but since all other parameters have been calibrated, its value is not relevant. 
\end{itemize}
This algorithm provides an easy-to-implement calibration of the rough Bergomi model using only information on the short-term regime. 

To test its accuracy, we let the \emph{market} parameters be $\sigma_0 = 0.2$, $H = 0.1$, $\eta = 0.5$ and $\rho = -0.7$, which we consider as ground truth. The algorithm described above yields the following calibration, which will be used in Section~\ref{numerical_rBergomi} as the benchmark for comparison with the signature-based model.

\begin{table}[H]
    \centering
    \begin{tabular}{ccc}
        \hline
        Parameter & True Value & Calibrated Value \\
        \hline
        $\sigma_0$ & 0.2 & 0.199884 \\
        $H$ & 0.1 & 0.100968 \\
        $\eta$ & 0.5 & 0.490527 \\
        $\rho$ & -0.7 & -0.672485 \\
        \hline
    \end{tabular}
    \caption{Calibrated parameters using the rough Bergomi calibration algorithm.}
    \label{tab:rBergomi-tab}
\end{table}

In Section \ref{SigBasedModel} we will introduce a data-driven model based on path signatures, which does not assume any specific parametric form for the volatility process. This approach can learn directly from a \emph{primary} noise, enabling it to adapt to a broader class of behaviors.

\section{Path Signatures} \label{pathsignatures}

A natural way to incorporate signatures into stochastic volatility modeling is through the framework proposed by \cite{cuchiero23}, where the \emph{asset price} is modeled as a linear functional of the signature of a driving noise process. Although this approach performs well under the assumption that the volatility process is a semimartingale, it is less suitable in settings characterized by rough volatility, where such regularity assumptions no longer hold.

To address this limitation, \cite{cuchiero25} propose an alternative formulation in which the \emph{volatility} process itself is expressed as a linear functional of the signature of the primary noise. Although this approach is computationally more intensive, it does not require the volatility to satisfy any martingale or semimartingale condition, making it particularly well-suited to the modeling of rough or highly irregular volatility dynamics.

In this paper, we adopt a similar approach, that is, we assume that the volatility is a continuous function of a general underlying stochastic process, called the \emph{primary noise}, which does not need to be of the Heston type. This continuous function is then approximated by a linear combination of the elements of the signature of the primary noise. 

We now introduce the essential ideas from rough path theory that underpin the signature-based approach. An insightful and clear exposition of rough paths is given in the Saint-Flour lecture notes by \cite{lyonscaruanalevy07}. Other good references are \cite{chevyrev16}, \cite{cuchiero23}, and \cite{lyons02}. We follow \cite{geng21} and \cite{perediaz23} in several places. We include some proofs to support intuition.

The need for signatures arises from the problem of defining integrals of the form
\begin{equation}\label{integralexample}
\int_s^t y_u\, dx_u
\end{equation}
when the integrand $y$ and integrator $x$ lack sufficient regularity. If both $x$ and $y$ have bounded variation, the integral is defined in the Riemann–Stieltjes or Lebesgue–Stieltjes sense. If $x$ and $y$ are $\alpha$-Hölder continuous with $\alpha > \frac{1}{2}$, Young's theory applies. However, when $\alpha \leq \frac{1}{2}$, classical constructions break down, and the Riemann sums
\[
\sum_{t_i \in \mathcal{D}} y_{t_{i-1}} (x_{t_i} - x_{t_{i-1}})
\]
may fail to converge as the mesh $|\mathcal{D}| \to 0$. At best, these sums provide a first-order approximation to the integral, and additional structure is needed to make sense of the limit.

Note that these approximations depend only on the increments $x_t - x_s$. In fact, the \emph{first level} of the signature of a path $x$ corresponds precisely to its increments. The signature can then be understood as an enhanced path that augments its first-order increments with higher-order information in the form of \emph{iterated integrals}.

To illustrate why higher-order terms are essential, we borrow the following example from \cite{geng21}. Consider a smooth function $F$ and let $y_t = F(x_t)$. Then, formally, one can write:
\begin{align*}
\int_s^t F(x_u)\, dx_u &= F(x_s)(x_t - x_s) + \int_s^t (F(x_u) - F(x_s))\, dx_u \\
&= F(x_s)(x_t - x_s) + \int_s^t \int_s^u DF(x_v)\, dx_v\, dx_u \\
&= F(x_s)(x_t - x_s) + DF(x_s)\int_s^t \int_s^u dx_v\, dx_u \\
&\quad + \int_s^t \int_s^u (DF(x_v) - DF(x_s))\, dx_v\, dx_u.
\end{align*}
Continuing recursively leads to the formal expansion
\begin{align*}
\int_s^t F(x_u)\, dx_u = F(x_s)(x_t - x_s) & + DF(x_s)\int_s^t \int_s^u dx_v\, dx_u \\
& + D^2F(x_s)\int_s^t \int_s^u \int_s^v dx_r\, dx_v\, dx_u \\
& + D^3F(x_s)\int_s^t \int_s^u \int_s^v \int_s^r dx_z\, dx_r\, dx_v\, dx_u + \cdots
\end{align*}
That is, computing the integral $\int_s^t F(x_u)\, dx_u$ requires access to the full collection of iterated integrals of $x$, not just its increments.

Note that if $x$ takes values in $\R^d$, then the second-level iterated integral
\[
\int_s^t \int_s^u dx_v\, dx_u
\]
is a tensor consisting of $d^2$ terms of the form $\int_s^t \int_s^u dx^i_v\, dx^j_u$. Higher-order levels live in higher tensor powers. Thus, the natural way to organize this structure is through the \emph{tensor algebra}, introduced formally below.

In low regularity settings (like Brownian motion or rough volatility models), these higher-order iterated integrals are not well-defined. Rough path theory allows us to \emph{define} them abstractly, thereby extending integration to paths of low regularity. 

Informally, if $x$ is $\alpha$-Hölder continuous, we expect that
\[
\left| \int_{s<t_1<\cdots<t_n<t} dx_{t_1} \cdots dx_{t_n} \right| \lesssim |t - s|^{n\alpha},
\]
so higher-order terms decay rapidly. This motivates approximating the integral of a function $F(x_t)$ against $dx_t$ as
\[
\int_s^t F(x_u)\, dx_u \ \approx \ \sum_{k=1}^N \text{DF}^{(k-1)}(x_s)\, \bX^k_{s,t},
\]
where $\bX^k_{s,t} = \int_{s<t_1<\cdots<t_k<t} dx_{t_1}\cdots dx_{t_k}$ will represent the $k$-th level of the signature of $x$. The truncation level $N$ will depend on the regularity of $x$. 

As a result, to define pathwise integration in irregular settings, and to model functionals of paths (such as volatility), we must specify a family of tensors $(\bX^k_{s,t})_{k=1}^N$ satisfying some algebraic and analytic constraints. These will form the signature of a rough path, which we now formalize by introducing the tensor algebra.

\subsection{Tensor Algebras}

Let $V$ be a real-valued finite dimensional vector space. In practice, $V$ will typically be $\R^d$, for some $d \geq 1$. For any non-negative integer $n$, we denote the $n$-th tensor power of $V$ as
\[
V^{\otimes n} := V \otimes \cdots \otimes V,
\]
with $V^{\otimes 0} = \R$. The tensor power of a vector space is a vector space as well. Moreover, if $V$ is isomorphic to $\R^d$ for some $d \geq 1$ then $V^{\otimes n}$ is isomorphic to $\R^{d^n}$. In particular, all tensor powers of $\R$ are isomorphic to $\R$ itself.  

If $e_1, \dots, e_d$ is a basis of $V$, then the elements $\{e_{i_1} \otimes \dots \otimes e_{i_n}; (i_1, \dots, i_n) \in \{1, \dots, d\}^n \}$ are a basis of $V^{\otimes n}$, that is, every tensor $v \in V^{\otimes n}$ can be written uniquely as
\[
v = \sum_{1 \leq i_1, \dots, i_n \leq d} \lambda_{i_1, \dots, i_n} e_{i_1} \otimes \dots \otimes e_{i_n},
\]
for some coefficients $\{\lambda_{i_1, \dots, i_n} \in \R; (i_1, \dots, i_n) \in \{1, \dots, d\}^n \}$. 

\begin{defi}[Extended Tensor Algebra]
    We define the \emph{extended tensor algebra $T((V))$ over $V$} as the set
    \[
    T((V)) = \{ \mathbf{a} = (a_0, a_1, \dots );\text{ } a_n \in V^{\otimes n} \}
    \]
    equipped with the following element-wise addition and scalar product
    \[
    \mathbf{a} + \mathbf{b} = (a_0 + b_0, \dots, a_n + b_n, \dots ), \quad \lambda \mathbf{a} = (\lambda a_0, \lambda a_1, \dots ),
    \]
    and endowed with the product $\otimes$ defined by
    \[
    \mathbf{a} \otimes \mathbf{b} = (c_0, c_1, \dots),
    \]
    where 
    \[
    c_n = \sum_{i+j = n} a_i \otimes b_j.
    \]
\end{defi}
In the same way that we can define a product in $T((V))$, we can characterize its invertible elements. Specifically, if $\mathbf{a} \in T((V))$ and the zeroth level $a_0 \in \R$ is nonzero, then $\mathbf{a}$ admits a multiplicative inverse in $T((V))$, given by the formal series:
\[
\mathbf{a}^{-1} = \sum_{n \geq 0} \frac{1}{a_0}\left( \1 - \frac{\mathbf{a}}{a_0} \right)^{\otimes n},
\]
where $\mathbf{1} := (1, 0, 0, \dots)$ is the multiplicative identity in $T((V))$, and the powers are taken with respect to the tensor product. Finally, we define the tensor algebra over $V$ as the set
\[
T(V) = \{ \mathbf{a} \in T((V)); \text{ } \exists n \in \mathbb{N} \text{ such that } a_k = 0 \text{ }\forall k \geq n\}.
\]
In other words, $T(V)$ consists of all formal tensor series with only finitely many nonzero terms. 

To make the notation more concise, define the multi-index $I = (i_1, \dots, i_n) \in \{1, \dots ,d\}^n$. We then write $e_I = e_{i_1} \otimes \cdots \otimes e_{i_n}$, and we denote the length of $I$ by $|I| = n$. In order to write scalars, we set $\alpha = \alpha e_{\emptyset}$, with $|\emptyset| = 0$. This notation allows us to write any tensor $v \in V^{\otimes n}$ as 
\[
v = \sum_{|I|=n} \lambda_I e_I,
\]
for some coefficients $\{\lambda_I \in \R; |I|=n\}$.

Given an element of $T(V)$, we can naturally associate a linear map on $T((V))$, in a manner analogous to the Riesz representation Theorem. 
\begin{defi}
    For any $\ell=\sum_{|I|\geq 0} \ell_I e_I \in T(V)$ and $\mathbf{a} = \sum_{|I|\geq 0} a_I e_I \in T((V))$, we define the map $\langle\cdot,\cdot\rangle:T(V) \times T((V)) \to \R$ by 
\begin{equation}
\langle \ell, \mathbf{a} \rangle := \sum_{|I| \geq 0} \ell_I a_I. \label{linearfunctional}
\end{equation}
\end{defi}
This map is well defined because there are only finitely many nonzero elements $\ell_I$. Note that we can recover the coordinate $a_I$ of $\mathbf{a}$ with $\langle e_I, \mathbf{a} \rangle = a_I$.  

We now introduce another important product on $T(V)$. The shuffle product is a way to combine two tensors in $T(V)$ by interweaving their entries in all possible ways, while \emph{preserving the relative order} within each tensor. Its effect is usually compared to that of shuffling cards from two decks while keeping each deck's internal order intact. The shuffle product is important in rough path theory because it encodes how products of iterated integrals combine.

\begin{defi} \label{shuffleprod}
    For any multi-indices $I = (i_1, \dots, i_n)$ and $J = (j_1, \dots, j_m)$, let $I' = (i_1, \dots, i_{n-1})$ and $J' = (j_1, \dots, j_{m-1})$. The shuffle product $e_I \shuffle e_J$ is defined recursively as
    \[
    e_I \shuffle e_J = (e_{I'} \shuffle e_J) \otimes e_{i_n} + (e_I \shuffle e_{J'}) \otimes e_{j_m},
    \]
    with the convention $e_I \shuffle e_{\emptyset} = e_{\emptyset} \shuffle e_I = e_I$.
\end{defi}
\begin{ex}
    If we consider $e_1 \otimes e_2$ and $e_3$ we get
    \[
    (e_1 \otimes e_2) \shuffle e_3 =  e_1 \otimes e_3 \otimes e_2 + e_3 \otimes e_1 \otimes e_2 + e_1 \otimes e_2 \otimes e_3.
    \]
\end{ex}

The shuffle product of two tensors of lengths $m$ and $n$ has $\binom{m + n}{m}$ elements. 
\begin{ex}\label{exShuffle2}
    Consider $I=\{1,2,3\}$ and $J=\{2,1\}$. With a slight abuse of notation, we may write $e_I=e_1\otimes e_2\otimes e_3=e_{123}$ and $e_J=e_2\otimes e_1=e_{21}$. To better observe the shuffling, in the expression below we underline the indexes corresponding to $e_{21}$:
    \begin{align*}
        e_{123} \shuffle e_{21} 
        &=  e_{123\underline{21}} + e_{\underline{2}123\underline{1}} + e_{\underline{21}123} + e_{\underline{2}1\underline{1}23} + e_{\underline{2}12\underline{1}3} \\
        &\quad + e_{1\underline{2}23\underline{1}} + e_{1\underline{2}2\underline{1}3} + e_{1\underline{21}23} + e_{12\underline{2}3\underline{1}} + e_{12\underline{21}3} \\
        &= e_{12321} + e_{12123} + 2e_{12213} + 2e_{12231} + 2e_{21123} + e_{21213} + e_{21231}.   
    \end{align*}
\end{ex}

While $e_1\otimes e_2\otimes e_3 \in V^{\otimes 3}$ and $e_2\otimes e_1\in V^{\otimes 2}$, note that $e_1\otimes e_2\otimes e_3 \shuffle e_2\otimes e_1 \in V^{\otimes 5}$. As we shall see in Section \ref{signatureappoximation}, the shuffle product sharply increases the order of computations required to construct the signature-based approximation to the volatility. 

Let $\ell^1, \ell^2\in T(V)$, with $\ell^1 =  \sum_{|I|\geq 0} \ell^1_I e_I$ and $\ell^2 =  \sum_{|J|\geq 0} \ell^2_J e_J$. We then have:
\[
\ell^1 \shuffle \ell^2 =  \sum_{|I|, |J| \geq 0} \ell^1_I \ell^2_J e_I \shuffle e_J.
\]
The collection $(T(V), +, \cdot, \shuffle)$ is a commutative algebra.

\subsection{Signature of Paths of Bounded Variation}\label{signatboundedvar}

We say that $\mathcal{D}_{[0,T]} = \{t_0, t_1, \dots, t_n\}$ is a partition of the interval $[0,T] $ if $0 = t_0 < t_1 < \cdots < t_n = T$. If the interval is clear from the context, we will simply write $\mathcal{D}$. Let $V$ be a $d$-dimensional vector space. 

\begin{defi}\label{p_variation}
    Let $p \geq 1$. A continuous path $X : [0,T] \to V$ has finite $p$-variation in $[0,T]$ if 
    \[
    ||X||_p = \left( \sup_{\mathcal{D}} \sum_{t_i \in \mathcal{D}} |X_{t_{i+1}} - X_{t_i}|^p\right)^{1/p}
    \]
    is finite. We denote by $\mathcal{V}^p([0,T])$ the set of continuous paths with finite $p$-variation in $[0, T]$.
\end{defi}
It is not difficult to show that if $1 < p < q$, then 
\[
\mathcal{V}^1([0,T]) \subset \mathcal{V}^p([0,T]) \subset \mathcal{V}^q([0,T]) \subset \mathcal{C}([0,T]).
\]
If $X \in \mathcal{V}^1([0,T])$, we say that $X$ is of bounded variation. 

To define the signature of a continuous path of bounded variation, we need to define integrals with respect to paths. Notice that, since $X$ has bounded variation, one can define the integral with respect to $X$ using Young's integration theory (see \cite{young}).

\begin{defi}[Signature]
    Consider a $d$-dimensional vector space $V$ and let $X:[0,T]\to V$ be a continuous path of bounded variation. Using the multi-index notation $I=(i_1,\dots,i_n)\in\{1,\dots, d\}^n$, we define the signature of $X$ on the interval $[0,T]$ as 
    \[
    S(X)_{0,T} = \sum_{|I|\geq 0} S(X)_{0,T}^I e_I,
    \]
    where the coefficients of the signature are defined recursively as 
    \begin{align*} 
    S(X)_{0,T}^0 & = \langle e_{\emptyset}, S(X)_{0,T}\rangle := 1 \nonumber \\
    S(X)_{0,T}^I & = \langle e_I, S(X)_{0,T}\rangle := \int_0^T \langle e_{I'}, S(X)_{0,s}\rangle dX^{i_{n}}_s. \label{signat_recursive}
    \end{align*}
\end{defi}

Recall that $I'=(i_1,\cdots,i_{n-1})$. The element of the signature corresponding to index $I$ can be written as 
\[
S(X)_{0,T}^{i_1, \dots, i_n} = \int_{0 < s_1 < \dots < s_n < T} dX_{s_1}^{i_1} \cdots dX_{s_n}^{i_n}.
\]

To gain some insight into the structure of the signature, consider the following two examples from \cite{chevyrev16}. 
\begin{ex}\label{exampleSig1}
Let $X:[0,t]\to\R$ be a one-dimensional path of bounded variation. Note that the multi-indexes are $I=(i_1,\cdots,i_n)\in\{1\}^n$. The signature of $X$ is then given by
\begin{align*} 
    S(X)_{0,t}^0 & = 1 \\
    S(X)_{0,t}^1 & = \int_0^t dX_s = X_t - X_0 \\
    S(X)_{0,t}^{11} & = \int_0^t \int_0^s dX_u dX_s = \frac{1}{2!}(X_t - X_0)^2 \\
    S(X)_{0,t}^{111} & = \int_0^t \int_0^s \int_0^u dX_rdX_u dX_s = \frac{1}{3!}(X_t - X_0)^3
\end{align*}
and so on. 
\end{ex}
\begin{ex}\label{exampleSig2}
Let $X:[0,5]\to\R^2$ be defined by $X_t = (X_t^1, X_t^2) = (3+t, (3+t)^2)$. Being a two-dimensional path, the multi-indexes are $I=(i_1,\cdots,i_n)\in\{1,2\}^n$. The elements of the signature are 
\allowdisplaybreaks
\begin{flalign*} 
    S(X)_{0,5}^0 & = 1 \\
    S(X)_{0,5}^1 & = \int_0^5 dX_t^1 = \int_0^5 dt = X_5^1 - X_0^1 = 5 \\
    S(X)_{0,5}^2 & = \int_0^5 dX_t^2 = \int_0^5 2(3+t)dt = X_5^2 - X_0^2 = 55 \\
    S(X)_{0,5}^{11} & = \int_0^5 \int_0^t dX_s^1 dX_t^1 = \int_0^5 \left[ \int_0^t ds \right] dt = \frac{25}{2} \\
    S(X)_{0,5}^{12} & = \int_0^5 \int_0^t dX_s^1 dX_t^2 = \int_0^5 \left[ \int_0^t ds \right] 2(3+t) dt = \frac{475}{3} \\
    S(X)_{0,5}^{21} & = \int_0^5 \int_0^t dX_s^2 dX_t^1 = \int_0^5 \left[ \int_0^t 2(3+s)ds \right] dt = \frac{350}{3} \\
    S(X)_{0,5}^{22} & = \int_0^5 \int_0^t dX_s^2 dX_t^2 = \int_0^5 \left[ \int_0^t 2(3+s)ds \right] 2(3+t)dt = \frac{3025}{2} \\
    S(X)_{0,5}^{111} & = \int_0^5 \int_0^t \int_0^s dX_u^1 dX_s^1 dX_t^1 = \int_0^5 \left[ \int_0^t \left[ \int_0^s du \right] ds \right] dt = \frac{125}{6}
\end{flalign*}
and so on. The signature of $X$ on $[0,5]$ can therefore be written as
\[
S(X)_{0,5} = (1,\ 5,\ 55,\ 12.5,\ 158.33,\ 116.66,\ 1512.5,\ 20.83,\dots).
\]
\end{ex}

The next result, known as Chen's identity, shows that, even though a path's signature is defined algebraically, it still captures the way the path evolves over time. In particular, it enables the reconstruction of the signature of $X$ over the interval $[0, T]$ provided that it is known on a collection of subintervals that cover $[0, T]$. 
\begin{teo}[Chen's identity]
    Let $X: [0,T] \to V$ be a continuous path of bounded variation. Then, for all $t\in (0, T)$,
    \begin{equation} \label{chenidentity}
    S(X)_{0,T} = S(X)_{0,t} \otimes S(X)_{t, T}.
    \end{equation}
\end{teo}
\begin{proof}
    We need to prove that, given $i_1, \dots, i_n$, then
    \[
    S(X)_{0,T}^{i_1, \dots, i_n} = \sum_{k = 0}^n S(X)_{0,t}^{i_1, \dots, i_k} S(X)_{t,T}^{i_{k+1}, \dots, i_n}
    \]
    We use induction on the level of signature $n$. If $n = 0$ then \eqref{chenidentity} is simply
    \[
    1 = 1 \otimes 1,
    \]
    which holds trivially. Assume that \eqref{chenidentity} holds for all $n \geq 0$. For $n+1$ we have
    \begin{align*}
        S(X)_{0,T}^{i_1, \dots, i_{n+1}} = & \int_0^T S(X)_{0,s}^{i_1, \dots, i_n} dX_s^{i_{n+1}}  \\
        = & \int_0^t S(X)_{0,s}^{i_1, \dots, i_n} dX_s^{i_{n+1}} + \int_t^T S(X)_{0,s}^{i_1, \dots, i_n} dX_s^{i_{n+1}} \\
        = &\ S(X)_{0,t}^{i_1, \dots, i_{n+1}} + \int_t^T S(X)_{0,s}^{i_1, \dots, i_n} dX_s^{i_{n+1}} \\
        = &\ S(X)_{0,t}^{i_1, \dots, i_{n+1}} + \int_t^T \sum_{k=0}^n S(X)^{i_1, \dots, i_k}_{0,t} S(X)_{t,s}^{i_{k+1}, \dots, i_n} dX^{i_{n+1}}_s,
    \end{align*}
    where the last equality follows from the induction step. Rearranging,
    \begin{align*}
        S(X)_{0,T}^{i_1, \dots, i_{n+1}} = &\ S(X)_{0,t}^{i_1, \dots, i_{n+1}} + \sum_{k=0}^nS(X)^{i_1, \dots, i_k}_{0,t}  \int_t^T S(X)_{t,s}^{i_{k+1}, \dots, i_n} dX^{i_{n+1}}_s \\
        = &\ S(X)_{0,t}^{i_1, \dots, i_{n+1}} + \sum_{k=0}^nS(X)^{i_1, \dots, i_k}_{0,t} S(X)_{t,T}^{i_{k+1}, \dots, i_{n+1}} \\
        = &\ \sum_{k = 0}^{n+1} S(X)^{i_1, \dots, i_k}_{0,t} S(X)_{t,T}^{i_{k+1}, \dots, i_{n+1}}.
    \end{align*}
which concludes the proof.
\end{proof}
From the definition of the signature, it is clear that $S(X)_{0,T} \in T((V))$. However, we can show that the signature of $X$ lies in a smaller space.
\begin{defi}
    An element $\mathbf{a} \in T((V))$ is said to be group-like if for every pair $\ell^1, \ell^2 \in T(V)$ we have
    \[
    \langle \ell^1, \bfa \rangle \langle \ell^2, \bfa \rangle = \langle \ell^1 \shuffle \ell^2, \bfa \rangle.
    \]
    We denote by $G(V)$ the set of group-like elements of $T((V))$.
\end{defi}

The group-like property is \emph{analogous} to the behavior of exponentials: just as $e^{x+y} = e^xe^y$ transforms addition into multiplication, signatures can be thought of as “exponentials” of paths rather than numbers. The shuffle product represents all the ways in which two tensors (say, $\ell_1$ and $\ell_2$) can be combined while preserving their internal order. The group-like condition says that evaluating the shuffle is equivalent to evaluating each tensor separately and multiplying the results. We now make this idea precise.

\begin{prop}\label{grouplikeprop}
    Let $X:[0,T] \to V$ be a continuous path of bounded variation. Then, the signature of $X$ satisfies the group-like property. That is, for every pair $\ell^1, \ell^2 \in T(V)$,  
    \[
    \langle \ell^1, S(X)_{0,T} \rangle \langle \ell^2, S(X)_{0,T} \rangle = \langle \ell^1 \shuffle \ell^2, S(X)_{0,T} \rangle.
    \]
\end{prop}
\begin{proof}
    By linearity, it is enough to prove it for $\ell^1 = e_I$ and $\ell^2 = e_J$. Let $n = |I| + |J|$. We will prove the result by induction on $n$. For $n=0$, we have $I=J=\emptyset$ and the result holds trivially. Assume that it holds for $n$, and let $I = (i_1, \dots, i_{n+1-m})$ and $J = (j_1, \dots, j_m)$. Note first that $\langle e_I, S(X)_{0,T} \rangle = S(X)_{0,T}^{I}$ and $\langle \e_J, S(X)_{0,T} \rangle = S(X)_{0,T}^{J}$. Using integration by parts and the notation for $I'$ and $J'$ introduced in Definition \ref{shuffleprod}, we have  
    \begin{align*}
        & \langle e_I, S(X)_{0,T} \rangle \langle e_J, S(X)_{0,T} \rangle \\
        & \quad = \int_0^T \langle e_J, S(X)_{0,s} \rangle d \langle e_I, S(X)_{0,s} \rangle + \int_0^T \langle e_I, S(X)_{0,s} \rangle d\langle e_J, S(X)_{0,s} \rangle   \\
        & \quad = \int_0^T \langle e_{I'}, S(X)_{0,s} \rangle \langle e_J, S(X)_{0,s} \rangle dX_s^{i_{n+1-m}}  + \int_0^T \langle e_I, S(X)_{0,s} \rangle \langle e_{J'}, S(X)_{0,s} \rangle dX_s^{j_m}.
    \end{align*}
    Using the induction step, it follows that 
    \begin{align*}
        & \langle e_I, S(X)_{0,T} \rangle \langle e_J, S(X)_{0,T} \rangle \\
        & \quad = \int_0^T \langle (e_{I'} \shuffle e_J), S(X)_{0,s} \rangle dX_s^{i_{n+1-m}}  + \int_0^T \langle (e_I \shuffle e_{J'}), S(X)_{0,s} \rangle dX_s^{j_m} \\
        & \quad = \langle (e_{I'} \shuffle e_J) \otimes e_{n+1-m}, S(X)_{0,T}\rangle  + \langle (e_I \shuffle e_{J'}) \otimes e_{j_m}, S(X)_{0,T} \rangle \\
        & \quad = \langle (e_I \shuffle e_J), S(X)_{0,T} \rangle, 
    \end{align*}
    which concludes the proof.
\end{proof}

The function $\langle\cdot,\cdot\rangle$ defined in (\ref{linearfunctional}) allows us to interpret the elements $\ell$ of the tensor algebra $T(V)$ as linear functionals when paired with a signature $S(X)_{0,T}$. If we evaluate two linear functionals $\ell_1$ and $\ell_2$ separately on the signature $S(X)_{0,T}$ and then multiply those two scalar values, the product equals what we get by evaluating the single functional $\ell^1 \shuffle \ell^2$ on $S(X)_{0,T}$. 

If we consider the path $X:[0,5]\to\R^2$ from Example \ref{exampleSig2}, we have 
\[
S(X)_{0,5}^{12} = \int_0^5 \int_0^t dX_s^1 dX_t^2 \quad \text{and} \quad S(X)_{0,5}^1 = \int_0^5 dX_t^1. 
\]
The multiplication of these two iterated integrals would be a polynomial in the components of the signature $S(X)_{0,5}$. The above proposition says that such a nonlinear expression can still be treated in a linear way provided we use the shuffle product. With the slight abuse of notation we used in Example \ref{exShuffle2}, we see that $e_{12} \shuffle e_1 = e_{121} + 2e_{112}$. Therefore, the product of the iterated integrals $S(X)_{0,5}^{12}$ and $S(X)_{0,5}^1$ can be expressed as a linear combination of 
\[
S(X)_{0,5}^{121} = \int_0^5 \int_0^t \int_0^s dX_u^1 dX_s^2 dX_t^1 \quad \text{and} \quad S(X)_{0,5}^{112} = \int_0^5 \int_0^t \int_0^s dX_u^1 dX_s^1 dX_t^2. 
\]

In other words, products of iterated integrals (polynomials in the elements of the signature) can be written as linear combinations of \emph{higher-order} integrals. The fact that the space of polynomials on signatures can be linearly organized via the shuffle product will be used in Section \ref{SigBasedModel}. The price to pay for linearity is the higher dimension of the tensor space in which the linear expression lives, which happens to be the space in which numerical computations will be carried out.

\subsection{Rough Paths} \label{RoughPaths}

So far we have developed the signature for continuous paths of bounded variation. By Young's integration theory, this construction extends to continuous paths of finite $p$-variation for $p < 2$. However, the stochastic processes most commonly used in finance---such as Brownian motion or fractional Brownian motion with small Hurst indexes---do not satisfy this condition. We therefore need to find a way to extend the signature to a broader class of paths with more irregular behavior.

This extension is achieved by \emph{lifting} the paths. Intuitively, \emph{lifting} refers to the process of enriching a path with additional information (namely, its iterated integrals). For smooth paths this yields the signature, while for more irregular paths this is done abstractly.

Let $X: [0,T] \to V$ be a path of finite $p$-variation for some $p\geq 2$, so that $X$ may be too irregular for classical iterated integrals to exist. To overcome this difficulty, we define a new object called a rough path 
\[
\bX_{s,t} = \left(\X^1_{s,t}, \X^2_{s,t}, \X^3_{s,t},\dots, \X^{\floor{p}}_{s,t} \right), 
\]
where $\X^1_{s,t}=X_t - X_s$ is the increment of the original path and each $\X^k_{s,t}$ is an approximation to the $k$-th order iterated integral
\[
\X^k_{s,t} \approx \int_{s<u_1<\cdots<u_k<t} dX_{u_1} \otimes\cdots \otimes dX_{u_k}
\]  
At this level of intuition, \emph{approximation} means that the $\X^k$ serve as proxies for the true iterated integrals, satisfying certain algebraic properties and appropriate $p$-variation bounds. This lifted structure enables us to define integration against $X$ even when classical approaches such as Riemann–Stieltjes or Young integration break down. In the remainder of this section we formalize these ideas.

\begin{defi}[Truncated Tensor Algebra]
Let $N\in \N$. We define the \emph{truncated tensor algebra of order $N$ over $V$} as
\[
T^N(V) = \{ \bfa = (a_0, a_1, \dots) \in T((V)); \text{ } a_k = 0 \text{ }\forall k > N\}.
\]
The projection map $\pi_{\leq N}: T((V))\to T^N(V)$ is defined as 
\[
\pi_{\leq N}((a_i)_{i=0}^{\infty}) = (a_i)_{i=0}^N.
\]
\end{defi}

For any $\bfa, \bfb\in T^N(V)$, we define the truncated tensor product in $T^N(V)$ as
\[
\bfa \otimes_{\leq N}\bfb = \pi_{\leq N}(\bfa\otimes\bfb).
\]
When dealing with elements of $T^N(V)$, if there is no risk of confusion we will generally use $\otimes$ to denote $\otimes_{\leq N}$.

Let $X: [0,T] \to V$ be a continuous path of bounded variation and $\Delta_{T} = \{(s,t) \in [0,T]^2 ; s \leq t\}$. The truncated signature of order $N$ of a path $X$ can therefore be defined as 
\begin{align*}
S(X)^{\leq N}: \Delta_T & \to T^N(V) \\
(s, t) & \mapsto \pi_{\leq N}(S(X)_{s,t}).
\end{align*}
We write $S(X)^{\leq N}(s,t) = S(X)^{\leq N}_{s,t}$. 

A slight modification in the proof of Chen's identity shows that, for all $0\leq s<u<t\leq T$,
\begin{equation}\label{multiplicativesignatN}
S(X)^{\leq N}_{s,t} = S(X)^{\leq N}_{s,u} \otimes S(X)^{\leq N}_{u,t}.
\end{equation}
We say that $S(X)^{\leq N}$ is \emph{multiplicative}. The following definition extends the multiplicative property to a more general setting. 

\begin{defi}[Multiplicative functional]
    For $N \in \N$, let $\bX : \Delta_T \to T^N(V)$ be a continuous map and denote $\bX(s,t)=\bX_{s,t}$. Since $\bX_{s,t} \in T^N(V)$, we can write $\bX_{s,t} = (\bX_{s,t}^0, \bX_{s,t}^1, \dots, \bX_{s,t}^N)$, where $\bX_{s,t}^k \in V^{\otimes k}$ for each $k$. We say that $\bX$ is a \textnormal{multiplicative functional} of degree $N$ in $V$ if, for every $(s,t)\in \Delta_T$, we have $\bX_{s,t}^0:=1$ and 
    \begin{equation} \label{chenrough}
    \bX_{s,t} = \bX_{s,u} \otimes \bX_{u,t}
    \end{equation}
    for all $s\leq u\leq t$.
\end{defi}

By extension, we also refer to \eqref{chenrough} as Chen's identity. Consider the case when $N = 1$. Chen's identity says that
\[
(1, \bX^{1}_{s,t}) = (1, \bX^1_{s,u}) \otimes (1, \bX^1_{u,t}) = (1, \bX^1_{s,u} + \bX^{1}_{u,t}),
\]
which implies $\bX^1_{s,t} = \bX^1_{s,u} + \bX^1_{u,t}$. The type of functionals that satisfy this property are called \emph{additive} functionals. Additivity provides an important step toward the definition of a rough path. 

Select an arbitrary constant $v\in V$ and define a path $\psi:[0,T]\to V$ by 
\[
\psi_t:=v + \bX_{0,t}^1.
\]
Then, $\psi_t - \psi_s = \bX_{0,t}^1 - \bX_{0,s}^1$. By additivity (Chen's identity at level $1$), this equals $\bX_{s,t}^1$. That is,
\[
\bX_{s,t}^1 = \psi_t - \psi_s.
\]
In other words, a multiplicative functional of order $1$ in $V$ is equivalent to the increment map of a path $\psi:[0,T]\to V$, unique up to an additive constant. 

Up to now, we started from a continuous path of finite variation $X: [0,T] \to V$ and showed that its truncated signature---an element of the tensor algebra $T^N(V)$---satisfies the multiplicative identity (\ref{multiplicativesignatN}). We now reverse this perspective: instead of constructing the signature from a classical path, we assume the algebraic structure of a signature and study its properties as a path taking values in the tensor algebra $T^N(V)$. 

We have already seen that the first level of a multiplicative functional is given by the increments of a path in $V$, just as the first level of the signature of a path corresponds to its own increments. We can now generalize this. 
\begin{lema}\label{additiveFunctionals}
    Let $\bX, \bY:\Delta_T\to T^N(V)$ be two multiplicative functionals of order $N$ that agree on the first $N-1$ levels. Then, the function $\Psi:\Delta_T\to V^{\otimes N}$ defined by $\Psi_{s,t} = \bX_{s,t}^N - \bY_{s,t}^N$ is additive, that is, $\Psi_{s,t} = \Psi_{s,u} + \Psi_{u,t}$, for all $s\leq u\leq t$.
\end{lema}
\begin{proof}
    Due to the multiplicative property, $\bX_{s,t} = \bX_{s,u}\otimes \bX_{u,t}$ for all $s\leq u\leq t$. Consider the $N$-th level component of the functionals in each side of the last equality. In the case of the $N$-th level component of $\bX_{s,u}\otimes \bX_{u,t}$, we separate the summands that only include elements of $V^{\otimes N}$ from the rest: 
    \[
    \bX^N_{s,t} = \bX^N_{s,u} + \bX^N_{u,t}  + \sum_{i+j=N-1} \bX^i_{s,u}\otimes \bX^j_{u,t}, 
    \] 
    where the elements of the summation are tensor products of elements from lower levels ($i,j\ge 0$). An analogous expression holds for $\bY_{s,t}^N$. Then,
    \[
    \Psi_{s,t} = \bX_{s,t}^N - \bY_{s,t}^N = \Psi_{s,u} + \Psi_{u,t} + \sum_{i+j=N-1} \left( \bX^i_{s,u}\otimes \bX^j_{u,t} - \bY^i_{s,u}\otimes \bY^j_{u,t} \right).
    \]
    As $\bX$ and $\bY$ agree on the first $N-1$ levels, the last term on the right-hand side is zero, yielding $\Psi_{s,t} = \Psi_{s,u} + \Psi_{u,t}$, which is what we needed to prove.
\end{proof}
\begin{lema}\label{lema2multiplicative}
    Let $\bX:\Delta_T\to T^N(V)$ be a multiplicative functional of order $N$ in $V$ and let $\Psi:\Delta_T\to V^{\otimes N}$ be an additive function. Then $\bX + \Psi$ is also a multiplicative functional. 
\end{lema}
\begin{proof}
    We need to show that, for all $s\leq u\leq t$,
    \[
    \bX_{s,t} + \Psi_{s,t} = (\bX_{s,u} + \Psi_{s,u}) \otimes  (\bX_{u,t} + \Psi_{u,t}).
    \]
    The right hand side of the above expression can be expanded as 
    \[
    \bX_{s,u}\otimes \bX_{u,t} + \bX_{s,u}\otimes \Psi_{u,t} + \Psi_{s,u}\otimes \bX_{u,t} + \Psi_{s,u}\otimes\Psi_{u,t}.
    \]
    As $\bX$ is multiplicative, $\bX_{s,u}\otimes \bX_{u,t} = \bX_{s,t}$. We now prove that  
    \begin{equation}\label{proofXpsi}
    \bX_{s,u}\otimes \Psi_{u,t} + \Psi_{s,u}\otimes \bX_{u,t} + \Psi_{s,u}\otimes\Psi_{u,t} = \Psi_{s,t}.
    \end{equation}
    Note that $\Psi$ takes values in $V^{\otimes N}$, which is the highest order component in the truncated tensor algebra $T^N(V)$. Therefore, when tensoring $\bX_{s,u}$ with $\Psi_{u,t}$, the only element in $V^{\otimes N}$ will be precisely $\Psi_{u,t}$. That is, $\bX_{s,u}\otimes_{\leq N} \Psi_{u,t} = \Psi_{u,t}$. More formally, assume that $\bX_{s,u}=\sum_{|J|\leq N}b_Je_J$ for some coefficients $b_J$, and let $I=(0,\dots,0,1)$ with $|I|=N$. Then, as $b_{\emptyset}=1$,  
    \[
    \bX_{s,u}\otimes \Psi_{u,t} = \left(\sum_{ |J|\leq N}b_Je_J \right)\otimes \Psi_{u,t}e_I = 1\Psi_{u,t}(e_{\emptyset}\otimes e_I) + \left(\sum_{|J|> N}c_Je_J \right),
    \]
    for some coefficients $c_J$, with $|J|>N$. We therefore have $\bX_{s,u}\otimes_{\leq N} \Psi_{u,t} = \Psi_{u,t}$. 
    The same reasoning applies to show that $\Psi_{s,u}\otimes_{\leq N} \bX_{u,t} = \Psi_{s,u}$. For the last term in left-hand side of (\ref{proofXpsi}), let $I$ be as above and $J=(0,\dots,0,1)$ with $|J|=N$. Then,  
    \[
    \Psi_{s,u}e_I \otimes \Psi_{u,t}e_J = \Psi_{s,u}\Psi_{u,t}e_I\otimes e_J.
    \]
    As $e_I\otimes e_J\in V^{\otimes2N}$, we have $\Psi_{s,u}\otimes_{\leq N} \Psi_{u,t}=0$. Writing $\otimes$ for $\otimes_{\leq N}$, it follows that  
    \[
    \bX_{s,u}\otimes \Psi_{u,t} + \Psi_{s,u}\otimes \bX_{u,t} + \Psi_{s,u}\otimes\Psi_{u,t} = \Psi_{s,u} + \Psi_{u,t} = \Psi_{s,t}, 
    \]
    where the last equation follows from the additivity assumption. This completes the proof.
\end{proof}

We now combine these results with the $p$-variation bounds on the functionals. Recall that a path $X:[0,T]\to V$ is $\alpha$-H\"older continuous if, for $s\leq t\in[0,T]$ and $0<\alpha\leq 1$, 
\[
|X_t-X_s| \leq C|t-s|^{\alpha} 
\]
for some constant $C>0$. If $X$ is $\alpha$-H\"older continuous, then it has finite $1/\alpha$-variation. The converse does not generally hold. However, if $X$ is a continuous path with finite $p$-variation, there exists a continuous, increasing reparametrization $\tau$ such that $X\circ\tau$ is $1/p$-H\"older continuous. 

For a continuous functional $\bX:\Delta_T\to T^N(V)$, we define 
\begin{equation}\label{p_var_norm}
    ||\bX||_{p\cdot \text{var}}:=\max_{1\leq k\leq N} \sup_{\mathcal{D}}{\left( \sum_{t_i\in\mathcal{D}} ||\pi_k(\bX_{t_i,t_{i+1}})||_{V^{\otimes k}}^{p/k} \right)^{k/p}},
\end{equation}
where the $\sup$ is taken over all the partitions $\mathcal{D}_{[0,T]}$. If $||\bX||_{p\cdot \text{var}}<\infty$, the functional $\bX$ is said to have finite $p$-variation. The $p$-variation distance between two functionals $\bX$ and $\bY$ of  finite $p$-variation is defined as 
\begin{equation}\label{pvardistance}
    d_{p\cdot \text{var}}(\bX,\bY):=||\bX-\bY||_{p\cdot \text{var}}.
\end{equation}

Assume that the multiplicative functionals $\bX, \bY:\Delta_T\to T^N(V)$ have finite $p$-variation and agree on the first $N-1$ levels. Then, by Lemma~\ref{additiveFunctionals}, the $N$-level difference,
\[
\Psi_{s,t}=\bX_{s,t}^N - \bY_{s,t}^N
\]
defines an additive function on $\Delta_T$. For a fixed $v\in V^{\otimes N}$, $\Psi$ induces a path $\psi:[0,T] \to V^{\otimes N}$ by
\[
\psi_t = v + \Psi_{0,t}.
\]
Then, by additivity of $\Psi$,
\[
\Psi_{s,t} = \psi_t - \psi_s 
\]
for any $s\leq t$. Additivity implies that we can think of $\Psi$ as a function that comes from the increments of a path. 

As $\bX$ and $\bY$ have finite $p$-variation, the difference $\Psi$ at level $N$ inherits finite $p/N$-variation. This follows directly from the structure of the $p$-variation norm (\ref{p_var_norm}), where the $N$-th level contributes with exponent $p/N$. 

It follows that there exists a continuous and increasing reparametrization $\tau$ of $[0,T]$ such that the reparametrized path $\Psi \circ \tau$ (defined by $\Psi_{\tau(s),\tau(t)} = \psi_{\tau(t)} - \psi_{\tau(s)}$) is $N/p$-H\"older continuous. This means that $\psi$ is regular enough to be treated as a genuine path in $V^{\otimes N}$, and that its increments $\Psi_{s,t} = \psi_t - \psi_s$, are well-behaved. Regularity allows us to reinterpret $\Psi$ as a “missing” top-level component that can be added to $\bY$ to produce a new functional
\[
\mathbf{Z}_{s,t}:=\bY_{s,t} + \Psi_{s,t},
\]
which, by Lemma~\ref{lema2multiplicative}, is multiplicative (and of order $N$).

Now, suppose that $N/p > 1$. Then, as $\psi \circ \tau$ is H\"older continuous with exponent greater than $1$, it must be constant. This implies that $\Psi_{s,t} = \psi_t - \psi_s = 0$ for all $s,t$, so the top level of $\bX$ and $\bY$ must coincide:
\[
\bX^N = \bY^N.
\]
Therefore, any two multiplicative functionals of finite $p$-variation that agree up to level $\floor{p}$ must in fact agree entirely. This suggests that the levels up to  $\floor{p}$ determine the rest.  

This observation raises the converse question: If a multiplicative functional $\bX$ is defined only up to level $N$, with finite $p$-variation and $N\geq \floor{p}$, can we extend it to higher levels in a consistent way? That is, can we construct a full multiplicative functional $\bY:\Delta_T\to T^n(V)$ with $n>N$, such that $\bY$ agrees with $\bX$ up to level $N$, and has finite $p$-variation?

Unlike the previous argument (no two extensions can differ when $N>\floor{p}$), this one is about existence and uniqueness of such an extension. The following result, proved in \cite{lyonscaruanalevy07}, answers the question affirmatively.
 
\begin{teo}[Extension Theorem]\label{extensionthm} 
Let $p \geq 1$ be a real number, $N \geq 1$ an integer, and let $\bX : \Delta_T \to T^N(V)$ be a multiplicative functional of degree $N$ with finite $p$-variation. Suppose that $N \geq \floor{p}$. Then, for every integer $n > N$, there exists a unique continuous multiplicative functional 
\[
\bY : \Delta_T \to T^n(V) 
\]
such that
\begin{itemize}
    \item [1.] $\bY$ agrees with $\bX$ up to level $N$; that is, $\pi_{\leq N}(\bY) = \bX$;
    \item [2.] $\bY$ has finite $p$-variation.
\end{itemize}
Moreover, the map that sends $\bX$ to its extension $\bY$ is continuous with respect to the $p$-variation metric.
\end{teo}
This result highlights that, for a multiplicative functional with finite $p$-variation, the first $\floor{p}$ levels completely capture all of its information, leading naturally to the next definition.

\begin{defi}[Rough Path]
    Let $p \geq 1$. A \emph{$p$-rough path} is a continuous multiplicative functional 
    \[
    \bX : \Delta_T \to T^{\floor{p}}(V)
    \]
    of degree $\floor{p}$ with finite $p$-variation. The space of $p$-rough paths is denoted by $\Omega^p_T(V)$.
\end{defi}

What distinguishes rough paths from general multiplicative functionals is that they retain only the minimal number of components necessary to capture all relevant information. In this sense, a rough path can be viewed as a “compressed” version of a multiplicative functional: it contains exactly the levels up to $\floor{p}$, which fully determine the rest under finite $p$-variation.

So far, we have not explicitly relied on the theory of signatures of bounded variation paths, except to define objects that reflect some of their structural properties. Rough paths form a highly abstract class, while paths of bounded variation are concrete and familiar. It is therefore natural to examine the rough paths that are \emph{close} to signatures of bounded variation paths, namely, those that arise as limits of such signatures. This leads us to the following definition.

\begin{defi}[Geometric Rough Paths] A \emph{geometric $p$-rough path} is a $p$-rough path $\bX$ for which there exists a sequence of paths of bounded variation $(X_n)_{n \geq 1}$ such that
\[
\lim_{n \to \infty} d_{p\cdot \textnormal{var}}(\bX, S(X_n)^{\leq \floor{p}}) = 0.
\]
The space of geometric $p$-rough paths is denoted by $G\Omega^p_T(V)$.
\end{defi}
Recall from Proposition \ref{grouplikeprop} that the signature of a bounded variation path satisfies the group-like property. It follows that each truncated signature $S(X_n)^{\leq \floor{p}}$ takes values in $G^{\floor{p}}(V)$, the set of group-like elements in the truncated tensor algebra $T^{\floor{p}}(V)$.

Now, since the group-like property is algebraic and preserved under limits, and the geometric $p$-rough path $\bX$ is defined as the limit of such signatures in the $p$-variation topology, it also takes values in $G^{\floor{p}}(V)$. Hence, every geometric $p$-rough path takes values in $G^{\floor{p}}(V)$.

Note, however, that the converse does not hold in general: not every $p$-rough path taking values in $G^{\floor{p}}(V)$ arises as the limit of signatures of bounded variation paths. This distinction motivates the following definition.

\begin{defi}[Weakly Geometric Rough Paths]
    A \emph{weakly geometric $p$-rough path} is a $p$-rough path taking values in $G^{\floor{p}}(V)$. The space of weakly geometric $p$-rough paths is denoted by $WG\Omega^p_T(V)$.
\end{defi}

The difference between $G\Omega^p_T(V)$ and $WG\Omega^p_T(V)$ is subtle and becomes relevant especially when $V$ is infinite-dimensional. It is always the case that
\[
G\Omega^p_T(V) \subset WG\Omega^p_T(V) \subset \Omega^p_T(V).
\]

\begin{ex}[Brownian Motion] \label{exBMotion}
As we often work with Brownian motion in the context of signature-based models, it is worth pausing to examine its signature and its interpretation as a rough path.

Let $B:[0,T] \to \R$ be a standard Brownian motion, and assume that stochastic integrals with respect to $B$ are defined in the Itô sense. Since Brownian motion has finite $p$-variation for any $p>2$, we can attempt to lift it to a $p$-rough path of degree $N=2$. The natural candidate for such a lift is the stochastic It\^o signature:
\begin{align*}
    S^{\text{It\^o}}(B)^{\leq 2}_{s,t} &= \left(1, \int_s^t dB_u,\; \int_s^t \int_s^u dB_r \, dB_u \right) \\
    &= \left(1, B_t - B_s,\; \int_s^t (B_u - B_s) \, dB_u \right) \\
    &= \left(1, B_t - B_s,\; \frac{1}{2}(B_t - B_s)^2 - \frac{1}{2}(t - s) \right),
\end{align*}
where the last identity follows from It\^o's formula, using $B_t^2 - B_s^2 = 2\int_s^t B_u \, dB_u + (t - s)$.

Now consider the shuffle identity. Since $e_1 \shuffle e_1 = 2e_1 \otimes e_1$, we have:
\[
\langle e_1, S^{\text{It\^o}}(B)^{\leq 2}_{s,t} \rangle^2 = (B_t - B_s)^2,
\]
but
\[
\langle e_1 \shuffle e_1, S^{\text{It\^o}}(B)^{\leq 2}_{s,t} \rangle = 2 \cdot \left( \frac{1}{2}(B_t - B_s)^2 - \frac{1}{2}(t - s) \right) = (B_t - B_s)^2 - (t - s).
\]
It follows that $S^{\text{It\^o}}(B)^{\leq 2}_{s,t}$ is not group-like, which means that It\^o integrals do not lead to weakly geometric rough paths. However, if we use Stratonovich integration, we obtain:
\begin{align*}
S^{\circ}(B)^{\leq 2}_{s,t} &= \left( 1, \int_s^t \circ dB_u,\; \int_s^t \int_s^u \circ dB_r \circ dB_u \right) \\
&= \left( 1, B_t - B_s,\; \frac{1}{2}(B_t - B_s)^2 \right).
\end{align*}
As $\langle e_1, S^{\circ}(B)^{\leq 2}_{s,t} \rangle^2$ matches 
\[
\langle e_1 \shuffle e_1, S^{\circ}(B)^{\leq 2}_{s,t} \rangle = 2 \cdot \frac{1}{2}(B_t - B_s)^2 = (B_t - B_s)^2,
\]
the Stratonovich signature is group-like and defines a weakly geometric $p$-rough path for any $p>2$.
\end{ex}

\medskip
\noindent Even though the It\^o signature is not group-like, we are not at a dead end. There are at least two standard approaches. One is to construct a weakly geometric rough path lift of Brownian motion by defining the second level as the Stratonovich iterated integral, i.e., work with
\[
\mathbf{B}_{s,t} := \left(1,\, B_t - B_s,\, \int_s^t (B_u - B_s) \circ dB_u \right),
\]
which is group-like. This is known as the \emph{Stratonovich lift} of Brownian motion and is the standard choice in rough path theory.

Alternatively, one could define a non-geometric rough path using the It\^o integral, but the path would lie outside the standard tensor algebra and would need to incorporate It\^o correction terms. This gives rise to \emph{branched} or \emph{generalized} rough paths (see \cite{brunedHZ21}). 

In general, the Stratonovich lift is preferred because it aligns with the algebraic structure of signatures (the group-like property) and allows for a direct interpretation of rough integrals as limits of classical Riemann–Stieltjes approximations. For one-dimensional paths, the following result shows that the situation is simpler.
\begin{lema}
Let $p\geq 1$ and let $X:[0,T] \to \R$ be a continuous path of finite $p$-variation. Then there exists a canonical lift to a weakly geometric $p$-rough path, given by
\[
\bX_{s,t} = \left(1, X_t - X_s, \frac{(X_t - X_s)^2}{2!}, \dots, \frac{(X_t - X_s)^{\floor{p}}}{\floor{p}!} \right).
\]
\end{lema}
\noindent To show that this expression defines a valid lift, note first that the right-hand side is the truncated exponential in the tensor algebra
\[
\exp^{\leq \lfloor p \rfloor}_{\otimes}(X_t - X_s) = \left(1, X_t - X_s, \frac{(X_t - X_s)^2}{2!}, \dots, \frac{(X_t - X_s)^{\floor{p}}}{\floor{p}!} \right),
\]
which is known to satisfy the Chen identity. We use the notation $\exp_{\otimes}$ to indicate that the exponential is taken in the tensor algebra.

Since $X$ has finite $p$-variation and each level $k$ is a smooth function of the increment $X_t - X_s$, the $k$-th level has finite $p/k$-variation. Hence, the full lift $\bX$ has finite $p$-variation in the sense of (\ref{p_var_norm}), and defines a weakly geometric $p$-rough path. 

In particular, for a one-dimensional Brownian motion $B:[0,T] \to \R$, the path
\[
\bB_{s,t} := \left( 1,\; B_t - B_s,\; \frac{1}{2}(B_t - B_s)^2 \right)
\]
defines a weakly geometric $p$-rough path for any $p>2$. This \emph{Stratonovich lift} of Brownian motion corresponds to the first two levels of the Stratonovich signature.

\begin{remark}\label{semimartingale_lift}
    More generally, any continuous semimartingale $S = A + M$, where $M$ is a continuous local martingale and $A$ is a continuous path of bounded variation on compact intervals, admits a canonical weakly geometric rough path lift (in the Stratonovich sense); see Chapter~14 of \cite{frizvictoir10}. The intuition is that the roughness of $S$ is driven by the martingale component $M$, while the bounded variation part $A$ can be handled using classical integration. The construction of the lift relies on probabilistic estimates, in particular the Burkholder-Davis-Gundy inequality and properties of the quadratic variation. 
\end{remark}

\subsection{Time-Augmented Rough Paths}

Recall that a \emph{lifted} path refers to extending a base path $X:[0,T]\to V$ to a rough path $\bX:\Delta_T\to T^N(V)$. We now proceed to \emph{augment} $X$ with a time coordinate. Specifically, we define $\hat{X} : [0,T] \to \R \oplus V$ by
\[
\hat{X}_t:= (t, X_t).
\]
At first glance, this change may appear superficial, but it has important consequences that we discuss at the end of this section. The following proposition says that time augmentation preserves important analytic and algebraic properties, such as admitting a lift to a weakly geometric $p$-rough path.
\begin{prop}\label{augmentXtohatX}
Let $X: [0,T] \to V$ be a continuous path that admits a weakly geometric $p$-rough path lift $\bX \in WG\Omega^p_T(V)$. Define the time-augmented path $\hat{X}:[0,T] \to \R \oplus V$ by
\[
\hat{X}_t := (t, X_t).
\]
Then $\hat{X} $ also admits a weakly geometric $p$-rough path lift $\hat{\bX} \in WG\Omega^p_T(\R \oplus V)$.
\end{prop}
\begin{proof}
We treat $\hat{X}$ as a path in $\R^{d+1}$, where the first coordinate is time and the remaining ones are given by $X$. We define the lifted path $\hat{\bX}$ inductively over multi-indices $I = (i_1, \dots, i_n) \in \{0, 1, \dots, d\}^n$, where $i_k = 0$ refers to the time component.

For $|I| = 1$, we set:
\[
\hat{\bX}^{(i)}_{s,t} := \begin{cases}
t - s & \text{if } i = 0, \\
X^i_t - X^i_s & \text{if } i \in \{1, \dots, d\}.
\end{cases}
\]

Assume that all levels up to length $|I|=n-1$ have been well defined. We now denote by $I$ the $n$-th level multi-index, $I = (i_1, \dots, i_n) \in \{0, 1, \dots, d\}^n$. The $n$-level components of the signature are 
\[
\langle e_I, \hat{\bX}_{s,t} \rangle = \hat{\bX}^I_{s,t}.
\]
If all elements of $I$ are nonzero ($i_k \neq 0, \text{ for } i=1,\dots,n$), we set 
\[
\hat{\bX}^I_{s,t} := \bX^I_{s,t}, 
\]
which is already defined in the original lift. Assume now that not all $i_k\neq 0$. Using the shuffle product introduced in Definition \ref{shuffleprod}, we have 
\[
e_I = e_{I'} \shuffle e_{i_n} - (e_{I''} \shuffle e_{i_n}) \otimes e_{i_{n-1}},
\]
where $I'$ has the same meaning as in Definition \ref{shuffleprod}, and $I''$ is defined in an analogous way. We then have
\begin{align}
    \langle e_I , \hat{\bX}_{s,t} \rangle = &\ \langle e_{I'} \shuffle e_{i_n} - (e_{I''} \shuffle e_{i_n}) \otimes e_{i_{n-1}}, \hat{\bX}_{s,t} \rangle \nonumber \\
    = &\ \langle e_{I'} \shuffle e_{i_n}, \hat{\bX}_{s,t}\rangle - \langle (e_{I''} \shuffle e_{i_n}) \otimes e_{i_{n-1}}, \hat{\bX}_{s,t}\rangle \nonumber \\
    = &\ \langle e_{I'}, \hat{\bX}_{s,t}\rangle \langle e_{i_n}, \hat{\bX}_{s,t}\rangle - \langle (e_{I''} \shuffle e_{i_n}) \otimes e_{i_{n-1}}, \hat{\bX}_{s,t}\rangle. \label{shuffle_augmented}
\end{align}
As $|I'|=n-1$ and $|(i_n)|=1$, it follows from the induction step that the terms $\langle e_{I'}, \hat{\bX}_{s,t}\rangle$ and $\langle e_{i_n}, \hat{\bX}_{s,t}\rangle$ are well defined. If $i_{n-1}=0$, define 
\[
\langle (e_{I''} \shuffle e_{i_n}) \otimes e_{i_{n-1}}, \hat{\bX}_{s,t}\rangle := \int_s^t \langle e_{I''} \shuffle e_{i_n} , \hat{\bX}_{s,u}\rangle du,
\]
where the integral is well-defined as a Young integral because $\bX$ is a $p$-rough path and $q=1$ for time, so that $1/p + 1/q > 1$ is always satisfied. If $i_{n-1}\neq 0$, using the shuffle product again we get:
\[
(e_{I''} \shuffle e_{i_n})\otimes e_{i_{n-1}} = e_{I''} \shuffle (e_{i_n}\otimes e_{i_{n-1}}) - (e_{I'''} \shuffle (e_{i_n}\otimes e_{i_{n-1}}))\otimes e_{i_{n-2}}.
\]
Therefore,
\begin{align}
    & \langle (e_{I''} \shuffle e_{i_n}) \otimes e_{i_{n-1}}, \hat{\bX}_{s,t}\rangle \nonumber \\ 
    & =\ \langle e_{I''}, \hat{\bX}_{s,t}\rangle \langle e_{i_n}\otimes e_{i_{n-1}}, \hat{\bX}_{s,t}\rangle - \langle (e_{I'''} \shuffle (e_{i_n}\otimes e_{i_{n-1}}))\otimes e_{i_{n-2}}, \hat{\bX}_{s,t}\rangle, \label{XhatRoughPathProof}
\end{align}
where the first term on the right hand side is well defined because of the induction step. If $i_{n-2}=0$, we define the second term on the right hand side as
\[
\langle (e_{I'''} \shuffle (e_{i_n}\otimes e_{i_{n-1}}))\otimes e_{i_{n-2}}, \hat{\bX}_{s,t}\rangle := \int_s^t \langle e_{I'''} \shuffle (e_{i_n}\otimes e_{i_{n-1}}), \hat{\bX}_{s,t}\rangle\,du,
\]
where, as above, the integral is well-defined as a Young integral. If $i_{n-2}\neq 0$, we apply shuffle product identity again to the second term on the right hand side of (\ref{XhatRoughPathProof}) to obtain a term involving $i_{n-3}$. We then repeat the process iteratively. 
\end{proof}

\begin{remark}
As a way of illustrating the construction above, let $I=\{i_1,i_2\}=\{0,1\}$. Since $I'=\{0\}$ and $I''=\emptyset$, we have
\[
(e_{I''} \shuffle e_{i_n})\otimes e_{i_{n-1}} = (e_{\emptyset} \shuffle e_{i_2})\otimes e_{i_1} = e_{i_2}\otimes e_{i_1} =  e_1\otimes e_0.
\]
It follows from (\ref{shuffle_augmented}) that 
\begin{equation} \label{exampleProof}
    \langle e_{\{0,1\}} , \hat{\bX}_{s,t} \rangle := \langle e_0, \hat{\bX}_{s,t}\rangle \langle e_1, \hat{\bX}_{s,t}\rangle - \langle e_1\otimes e_0, \hat{\bX}_{s,t}\rangle 
\end{equation}
The left hand side of the above expression is:
\[
\langle e_{\{0,1\}} , \hat{\bX}_{s,t} \rangle = \int_s^t \int_s^u dr\, dX_u = \int_s^t (u-s) dX_u,
\]
which, integrating by parts, yields
\begin{equation} \label{LHSproof}
\langle e_{\{0,1\}} , \hat{\bX}_{s,t} \rangle = (t-s)(X_t-X_s) - \int_s^t (X_u-X_s)du.
\end{equation}
The elements of the signature with indexes of length $1$ were defined as $\langle e_0, \hat{\bX}_{s,t}\rangle := t-s$ and $\langle e_1, \hat{\bX}_{s,t}\rangle := X_t-X_s$. Therefore, the right hand side of (\ref{exampleProof}) is:
\begin{align*}
\langle e_0, \hat{\bX}_{s,t}\rangle \langle e_1, \hat{\bX}_{s,t}\rangle - \langle e_1\otimes e_0, \hat{\bX}_{s,t}\rangle
& = (t-s)(X_t-X_s) - \int_s^t \int_s^u dX_r\,du \\
& = (t-s)(X_t-X_s) - \int_s^t (X_u-X_s) du,
\end{align*}
which is the same as (\ref{LHSproof}).
\end{remark}

Now that we know that there exists a lift for the time-augmented path $\hat{X}$, we define the corresponding space. 

\begin{defi}[Time-Augmented Weakly Geometric Rough Path]\label{DEFaugmentedWGpRP}
A \emph{time-augmented weakly geometric $p$-rough path} is a weakly geometric $p$-rough path $\bX \in WG\Omega^p_T(\R \oplus V)$ such that
\begin{itemize}
    \item [(i)] the first level satisfies $\pi_1(\bX_{s,t}) = \hat{X}_t - \hat{X}_s$, where $\hat{X}_t = (t, X_t)$ for some continuous path $X: [0,T] \to V$ that admits a weakly geometric $p$-rough path lift;
    \item [(ii)] for any $I$ with $|I|\leq \lfloor p\rfloor$,
                \[
                \langle e_{I0},\bX_{s,t} \rangle = \int_s^t \langle e_I, \bX_{s,u} \rangle du,
                \] 
    where the integral is a Young integral, and where the notation $e_{I0}$ means appending a $0$ (the time component) to the multi-index $I$.
\end{itemize}
We denote the space of such paths by $WG\hat{\Omega}^p_T(V)$.
\end{defi}

\begin{remark}
    A rough path in $WG\hat{\Omega}^p_T(V)$ presupposes the existence of a path $X:[0,T] \to V$, which is augmented to $\hat{X}:[0,T] \to \R \oplus V$. This contrasts with the rough paths in $WG\Omega^p_T(V)$, which are defined independent of any base path.
\end{remark}

\begin{remark}
    Being a weakly geometric rough path in $\R \oplus V$ with the correct first level does not automatically guarantee that the component $\langle e_{I0}, \bX_{s,t} \rangle$ should equal the integral of $\langle e_I, \bX_{s,u} \rangle$ against the time increment $du$. This why we need to include condition ($ii$). The Young integral is well defined since time has bounded variation. 
\end{remark}

We now briefly discuss the relevance of lifting the augmented path $(t, X_t)$ instead of lifting $X_t$ alone. Let $S(X)$ denote the signature of a path $X : [0,T] \to V$. It is well known that the map $X \mapsto S(X)$ is not injective: different paths can share the same signature—for example, if they trace out the same image at different speeds. Thus, the signature loses information about the timing or parametrization of the path.

Augmenting the path by adding time, that is, lifting $\hat{X}_t:= (t, X_t) \in \R\oplus V$, recovers injectivity. Theorem~\ref{signatuniqueness} below shows that the time-augmented signature uniquely determines the path, making it more expressive. 

Time augmentation allows the signature to distinguish between paths that follow the same geometric shape but evolve at different speeds. This distinction is crucial in the context of the universal approximation theorems discussed below: time augmentation ensures that the signature map separates paths in a sufficiently rich way to apply approximation results such as the Stone--Weierstrass theorem.

As we will also see below, time augmentation is essential for learning path-dependent functionals in an \emph{adapted} way. In models where $Y_t$ depends on the past trajectory $(X_s)_{s \leq t}$, incorporating time allows the signature to detect \emph{when} events occur. Without time, the signature treats two paths with the same shape but different timing as indistinguishable—an undesirable feature in many stochastic or temporally sensitive learning tasks.

\subsection{Signatures of Rough Paths}

Let $\bX \in \Omega^p_T(V)$ be a $p$-rough path. By the extension theorem, for every integer $N \geq \floor{p}$, there exists a unique multiplicative extension of $\bX$ to degree $N$ with finite $p$-variation. Since this extension process can be carried out to arbitrarily high degrees, it is natural to define the \emph{signature} of a $p$-rough path as its formal infinite extension. This motivates the following definition.
\begin{defi}[Signature of a Rough Path]\label{signatRoughPath}
Let $\bX \in \Omega^p_T(V)$ be a $p$-rough path. The \emph{truncated signature of order $N \geq \floor{p}$} is defined as the unique extension of $\bX$ to level $N$ with finite $p$-variation, denoted by:
    \[
    S(\bX)^{\leq N} := \left(1, S(\bX)^1, \dots, S(\bX)^N \right) \in T^N(V).
    \]
    The \emph{(full) signature} of $\bX$ is the formal series
    \begin{align*}
    S(\bX): \Delta_T & \to T((V)) \\
    (s,t) & \mapsto S(\bX)_{s,t} := (1, S(\bX)^1_{s,t},\dots,S(\bX)^n_{s,t},\dots),
    \end{align*}
    where $S(\bX)^n_{s,t} \in V^{\otimes n}$ denotes the $n$-th level of the extension.
\end{defi}

Let $\bX \in WG\Omega^p_T(V)$ and $1<p<q$. In \cite{frizvictoir10} it is proved that 
\[
WG\Omega^p_T(V) \subset G\Omega^q_T(V).
\]
Therefore, there exists a sequence of bounded variation paths $X^{(n)}$ such that their truncated signatures of order $\lfloor q \rfloor$ converge to $\bX$ in the $q$-variation topology. By Theorem \ref{extensionthm}, the extension map $\bX \mapsto S(\bX)^{\leq \floor{q}}$ is continuous in the $q$-variation topology. Since each $S(X^{(n)})^{\leq \floor{q}}$ is group-like, it follows that $S(\bX)^{\leq \floor{q}}$ is group-like as well. Letting $q \to \infty$, we obtain that the full signature $S(\bX) \in T((V))$ is group-like. That is, for any $\ell^1, \ell^2 \in T(V)$, we have:
\[
\langle\ell^1, S(\bX)_{s,t}\rangle \langle\ell^2, S(\bX)_{s,t}\rangle = \langle\ell^1 \shuffle \ell^2, S(\bX)_{s,t}\rangle.
\]

\noindent Note the slight shift in notation that has taken place here. In Section~\ref{signatboundedvar}, we began with a path $X : [0,T] \to V$ and constructed its signature, denoted by $S(X) : \Delta_T \to T((V))$. In contrast, we now start with a multiplicative functional $\mathbf{X} : \Delta_T \to T^{\lfloor p \rfloor}(V)$, which satisfies specific algebraic and analytic properties (multiplicativity and finite $p$-variation) that allow it to be uniquely extended to higher levels $N \geq \lfloor p \rfloor$. We can describe this extension using the diagram:
\[
\begin{array}{ccccc}
\bX^{\floor{p}}(V) & \longrightarrow & \bX^N(V) & \longrightarrow & \bX^\infty \\[1.3ex]
(1, \bX^1_{s,t}, \dots, \bX^{\lfloor p \rfloor}_{s,t}) & \longrightarrow & 
(1, \bX^1_{s,t}, \dots, \bX^N_{s,t}) & \longrightarrow &
(1, \bX^1_{s,t}, \dots, \bX^n_{s,t}, \dots)
\end{array}
\]
where we slightly abuse notation in referring to $\mathbf{X}^\infty$ as the infinite extension of the rough path.

In Definition~\ref{signatRoughPath}, however, we formalized this extension using the signature notation:
\[
\begin{array}{ccccc}
\bX^{\floor{p}}(V) & \longrightarrow & S(\bX)^{\leq N}(V) & \longrightarrow & S(\mathbf{X}) \\[1.3ex]
(1, \bX^1_{s,t}, \dots, \bX^{\lfloor p \rfloor}_{s,t}) & \longrightarrow & 
(1, S(\bX)^1_{s,t}, \dots, S(\bX)^N_{s,t}) & \longrightarrow &
(1, S(\bX)^1_{s,t}, \dots, S(\bX)^n_{s,t}, \dots)
\end{array}
\]

Note carefully that we have not \emph{constructed a signature} from $\mathbf{X}$. Rather, we started with a multiplicative functional satisfying the group-like property and finite $p$-variation (a $p$-rough path), and we uniquely extended it from level $\lfloor p\rfloor$ to higher levels in the (truncated) tensor algebra. This extension is granted by Theorem~\ref{extensionthm}. In Definition~\ref{signatRoughPath}, we have simply relabeled this extended functional using the signature notation previously introduced for bounded variation paths.

Although this notation may initially appear inconsistent, it aligns naturally with the classical case where the signature is defined via iterated integrals. This unification of notation allows both settings—bounded variation and rough paths—to be treated within a common framework.

\subsection{Universal Approximation Theorems}

Since the primary role of a rough path $\bX$ is to serve as a driver for integrals (or differential equations), the key object of interest is the increment $\bX_{s,t}$ rather than the individual value $\bX_t$ at a fixed time. But the two viewpoints are equivalent. Given the increments $\bX_{s,t}$, one can reconstruct a path by defining $\bX_t := \bX_{0,t}$. If the path $\bX_t \in T^N(V)$ is given, the increments can be recovered by making use of the multiplicative property, namely, $\bX_{s,t} := \bX_s^{-1} \otimes \bX_t$.

In this section, we present three fundamental theorems that enable the application of signatures to the calibration problem. The first one states that the signature (as defined in Definition~\ref{signatRoughPath}) of a time-augmented weakly geometric $p$-rough path $\bX$, evaluated at time $T$, uniquely determines $\bX$. We use the simplified notation $\bX_t:= \bX_{0,t}$.
\begin{teo}[Uniqueness of the Signature] \label{signatuniqueness}
    Let $\bX$, $\bY \in WG\hat{\Omega}^p_T(V)$. Then,  
    \[
    S(\bX)_T = S(\bY)_T \iff \forall t \in [0,T],\ \bX_t = \bY_t. 
    \]
\end{teo}
\begin{proof}
    Note first that $\bX:\Delta_T\to T^{\floor{p}}(\R\oplus V)$, while $S(\bX):\Delta_T\to T((\R\oplus V))$. If $\bX_t = \bY_t$ for all $t \in [0,T]$, it must be that $\bX_T = \bY_T$. By the extension theorem, $S(\bX)_T  = S(\bY)_T$. 
    
    \noindent Reciprocally, assume $S(\bX)_T = S(\mathbf{Y})_T $. Since $\bX, \bY \in WG\hat{\Omega}^p_T(V)$ and both coincide with their signatures up to level $\floor{p}$, it suffices to show that for any multi-index $I$ with $|I| \leq \floor{p}$, 
    \[
    \langle e_I , \bX_t \rangle = \langle e_I, \bY_t \rangle \quad \text{for all } t \in [0,T].
    \]
    To extract these values from the full signature $S(\bX)_T$, we face the difficulty that $\langle e_I, \bX_t \rangle$ is a function of time, not a coefficient in the signature at time $T$. However, we can recover such functions by integrating them against powers of time. To do this, we first prove that:
    \[
    \langle e_0^{\otimes k}, \bX_t \rangle = \frac{t^k}{k!} \quad \text{for all } k \geq 0.
    \]
    We proceed by induction. For $k = 0$, the claim is trivial. For $k = 1$, we have $\langle e_0, \bX_t \rangle = t$ by the definition of time augmentation. Assuming the identity holds for $k-1$, we get:
    \[
    \langle e_0^{\otimes k}, \bX_t \rangle = \int_0^t \langle e_0^{\otimes (k-1)}, \bX_u \rangle \, du = \int_0^t \frac{u^{k-1}}{(k-1)!} \, du = \frac{t^k}{k!}.
    \]
    Note that while condition ($ii$) in Definition \ref{DEFaugmentedWGpRP} applies to $k\leq \lfloor p \rfloor$, the relation extends to all $k\geq 0$ due to the bounded variation of the time component, the way $\bX \in WG\hat{\Omega}^p_T(V)$ is constructed (Proposition \ref{augmentXtohatX}) and the uniqueness of the signature extension.
    
    Now fix any multi-index $I$ with $|I| \leq \floor{p}$. Since $\bX_t$ and $S(\bX)_t$ agree up to level $\floor{p}$, we have
    \[
    \langle e_I, S(\bX)_t \rangle = \langle e_I, \bX_t \rangle.
    \]
    Using the group-like property and the shuffle product identity, we have
    \begin{align*}
    \langle (e_I \shuffle e_0^{\otimes k}) \otimes e_0, S(\bX)_T \rangle 
    &= \int_0^T \langle e_I \shuffle e_0^{\otimes k}, S(\bX)_u \rangle \, du \\
    &= \int_0^T \langle e_I, S(\bX)_u \rangle \langle e_0^{\otimes k}, S(\bX)_u \rangle \, du \\
    &= \int_0^T \langle e_I, \bX_u \rangle \frac{u^k}{k!} \, du.
    \end{align*}
    Applying the same reasoning to $\bY$, and using the assumption $S(\bX)_T = S(\bY)_T$, we conclude that:
    \[
    \int_0^T \left( \langle e_I, \bX_u \rangle - \langle e_I, \bY_u \rangle \right) \frac{u^k}{k!} \, du = 0.
    \] 
        
    Since the functions $u \mapsto \frac{u^k}{k!}$ (for $k \geq 0$) form a basis for the space of polynomials, and polynomials are dense in $C([0,T])$ by the Stone–Weierstrass theorem, it follows that any continuous function whose integral against every such monomial vanishes must be identically zero. As $\bX_u$ and $\bY_u$ are continuous, it follows that 
    \[
    \langle e_I, \bX_u \rangle = \langle e_I, \bY_u \rangle \quad \text{ for all } u \in [0,T],
    \]
    which completes the proof.
\end{proof}

Before proceeding, note that $WG\hat{\Omega}^p_T(V)$, equipped with the $p$-variation distance, becomes a topological space whose topology is induced by the $p$-variation metric.

\begin{teo}[First Universal Approximation Theorem] Let $K \subset WG\hat{\Omega}^p_T(V)$ be compact, and let $f:WG\hat{\Omega}^p_T(V)\to \R$ be continuous with respect to the $p$-variation topology. Then, for every $\eps > 0$, there exists $\ell \in T(\R\oplus V)$ such that
\[
\sup_{\bX \in K} |f(\bX) - \langle \ell, S(\bX)_T \rangle | < \eps.
\]
\end{teo}
\begin{proof}
We apply the Stone–Weierstrass theorem to a suitable subalgebra of continuous functions on the compact set $K \subset WG\hat{\Omega}^p_T(V)$. Define
\[
A := \operatorname{span} \left\{ \bX \mapsto \langle e_I, S(\bX)_T \rangle;\, I \text{ multi-index in } \R \oplus V \right\} \subset C(K).
\]
That is, $A$ is the collection of all finite linear combinations of coordinate functionals on the signature evaluated at time $T$. 

To show that $A$ is a subalgebra of $C(K)$ that satisfies the conditions of the Stone–Weierstrass theorem, consider first that $A$ is closed under multiplication. As the group-like property  
\[
\langle \ell_1 \shuffle \ell_2, S(\bX)_T \rangle = \langle \ell_1, S(\bX)_T \rangle  \langle \ell_2, S(\bX)_T \rangle
\]
corresponds to multiplication of linear functionals on the signature, the span of such functionals is closed under multiplication. 

Second, $A$ separates points in $K$ due to the uniqueness of the signature of time-augmented weakly geometric rough paths: in particular, Theorem~\ref{signatuniqueness} guarantees that if $\bX \neq \bY$ in $K$, then $S(\bX)_T \neq S(\bY)_T$, so there exists $\ell \in T(\R \oplus V)$ such that $\langle \ell, S(\bX)_T \rangle \neq \langle \ell, S(\bY)_T \rangle$. 

And third, $A$ contains the constant functions, as can be seen by choosing $I = \emptyset$, for which $\langle e_I, S(\bX)_T \rangle = 1$. Therefore, by the Stone–Weierstrass theorem, $A$ is dense in $C(K)$. In particular, for any $\varepsilon > 0$, there exists $\ell \in T(\R \oplus V)$ such that
\[
\sup_{\bX \in K} \left| f(\bX) - \langle \ell, S(\bX)_T \rangle \right| < \varepsilon.
\]
\end{proof}
This result is important. The signature $S(\bX)$, evaluated at time $T$, serves as a feature map that transforms a path into an infinite sequence of coordinates capturing all relevant information. Theorem~\ref{signatuniqueness} ensures that this representation is injective for time-augmented paths. It is then natural to ask whether signatures are rich enough to approximate functionals defined on paths. The First Universal Approximation Theorem answers this question affirmatively.

It specifically states that continuous functionals on compact subsets of the rough path space can be approximated arbitrarily well by linear functionals on the signature—that is, by finite linear combinations of iterated integrals. This makes signatures a powerful tool for representing and learning functionals on paths, especially in contexts such as calibration or supervised learning.

Our goal is to model one-dimensional stochastic processes of the form
\[
Y_t = f\big((\bX_s)_{s \in [0,t]}\big),
\]
where $(\bX_s)_{s \in [0,t]}$ is a stochastic process, with each $\bX_s \in WG\hat{\Omega}^p_t(V)$. Denote this stochastic process by $\mathcal{X}$:
\[
\begin{array}{rccc}
\mathcal{X}: & [0,t] & \longrightarrow & WG\hat{\Omega}_t^p(V) \\ [1ex] 
             & s & \longmapsto & \bX_s.
\end{array}
\]
Think of $Y_t$ as a volatility process driven by a rough signal $\mathcal{X}$. The function $f$, which encodes this dependence, is typically unknown.

Note that, as time increases, the domain of $f$ changes: for each $t$, $(\bX_s)_{s \in [0,t]}$ represents a collection of rough paths defined on a different time interval. To deal with this, we need a consistent way to interpret a rough path in $WG\hat{\Omega}^p_u(V)$ as an element of $WG\hat{\Omega}^p_t(V)$ for all $t \geq u$, so that $f$ can act on a common space. 

Given the potential complexity of the notation, we begin with a simple case and build up from there. Let the weakly geometric $p$-rough path $\bX \in WG\hat{\Omega}^p_T(V)$ be represented by the diagram:
\[
\begin{array}{ccl}
\Delta_T & \longrightarrow & T^{\floor{p}}(\R \oplus V) \\[1ex]
(s,t) & \longmapsto & \bX_{s,t} = \left(1, \bX^1_{s,t}, \dots, \bX^{\floor{p}}_{s,t} \right),
\end{array}
\]
for $0 \leq s \leq t \leq T$. Fixing $s=0$, we write $\bX_t := \bX_{0,t}$.

By Proposition~\ref{augmentXtohatX}, a rough path in $ WG\hat{\Omega}^p_s(V)$ originates from a base path $X:[0,s]\to V$ that admits a lift $\bX \in WG\Omega^p_s(V)$. To emphasize that this rough path lives on $[0,s]$, we denote it by $_{[\mathsf{s}]}\bX$, and it is represented by 
\[
\begin{array}{ccl}
\Delta_s & \longrightarrow & T^{\floor{p}}(V) \\[1ex]
(r,u) & \longmapsto & _{[\mathsf{s}]}\bX_{r,u} = \left(1,\, _{[\mathsf{s}]}\bX^1_{r,u}, \dots,\, _{[\mathsf{s}]}\bX^{\floor{p}}_{r,u} \right).
\end{array}
\]
For $u \in [0,s]$, we write $_{[\mathsf{s}]}\bX_u :=\! _{[\mathsf{s}]}\bX_{0,u}$.

The time-augmented path is $\hat{X}_u := (u, X_u)$, $u \in [0,s]$. By Proposition~\ref{augmentXtohatX}, $\hat{X}$ admits a weakly geometric $p$-rough path lift $_{[\mathsf{s}]}\hat{\bX} \in WG\hat{\Omega}^p_s(V)$. We refer to $_{[\mathsf{s}]}\hat{\bX}$ as the stopped path at time $s$.  

To extend this construction from $[0,s]$ to $[0,t]$, with $s<t$, we define $_{[\mathsf{t}]}\bX \in WG\Omega^p_t(V)$ by:
\begin{equation}\label{timeextension_st}
_{[\mathsf{t}]}\bX_u :=
\begin{cases}
_{[\mathsf{s}]}\bX_u & \text{for } u \in [0,s] \\
_{[\mathsf{s}]}\bX_s & \text{for } u \in [s,t].
\end{cases}
\end{equation}
We then apply the construction from the proof of Proposition~\ref{augmentXtohatX} to obtain a time-augmented rough path $_{[\mathsf{t}]}\hat{\bX} \in WG\hat{\Omega}^p_t(V)$. By construction, $_{[\mathsf{t}]}\hat{\bX}$ agrees with $_{[\mathsf{s}]}\hat{\bX}$ on $[0,s]$.

This provides a consistent way to extend truncated rough paths defined on $[0,s]$ to the full interval $[0,t]$, thus allowing $f$ to act on a unified space. The ability to extend time-augmented paths motivates the following definition.

\begin{defi}[Stopped Rough Path] 
Let $p \geq 1$. We define the space of \emph{weakly geometric stopped $p$-rough paths} as
\[
\Lambda^p_T(V) := \bigcup_{t \in [0,T]} WG\hat{\Omega}^p_t(V).
\]
Given $_{[\mathsf{s}]}\bX \in WG\hat{\Omega}^p_s(V)$ and $_{[\mathsf{t}]}\!\bY \in WG\hat{\Omega}^p_t(V)$, with $s \leq t$, we define a metric on $\Lambda^p_T(V)$ by
\[
d( _{[\mathsf{s}]}\bX,\, _{[\mathsf{t}]}\!\bY ) := d_{p\cdot \textnormal{var}} (_{[\mathsf{t}]}\bX,\, _{[\mathsf{t}]}\!\bY ) + |t - s|,
\]
where $_{[\mathsf{t}]}\bX$ denotes the extension of $_{[\mathsf{s}]}\bX$ from $[0,s]$ to $[0,t]$ as in (\ref{timeextension_st}), and $d_{p\cdot \textnormal{var}}$ denotes the $p$-variation distance on $WG\hat{\Omega}^p_t(V)$.
\end{defi}

\noindent
For further details on the topology of this space, see \cite{kalsi20} and \cite{bayer23}. The concept of stopped rough paths provides a useful framework for handling adaptedness in stochastic settings. The Second Universal Approximation Theorem below is formulated in this space.

\begin{teo}[Second Universal Approximation Theorem]
Let $K \subset WG\hat{\Omega}^p_T(V)$ be a compact subset, and let $f : \Lambda^p_T(V) \to \R$ be a continuous function. Then, for every $\varepsilon > 0$, there exists $\ell \in T(\R \oplus V)$ such that
\[
\sup_{\bX \in K,\; t \in [0,T]} \left| f(_{[\mathsf{t}]}\bX) - \langle \ell, S(_{[\mathsf{t}]}\bX)_t \rangle \right| < \varepsilon,
\]
where $_{[\mathsf{t}]}\bX$ denotes the restriction of $\bX$ to $[0,t]$, and $S(_{[\mathsf{t}]}\bX)_t := S(_{[\mathsf{t}]}\bX)_{0,t}$ is the signature of the stopped rough path up to time $t$.
\end{teo}
\begin{proof}
    The proof is also based on the Stone-Weierstrass theorem, but is somewhat more technical. See \cite{kalsi20}, Lemma~B.3.
\end{proof}

This result shows that any continuous functional of a stopped rough path can be uniformly approximated, over all truncating times $t\in [0,T]$ and all paths in a compact set, by a linear functional on the signature. The approximation is uniform both in path space and in time.

\section{Signature-Based Volatility Models} \label{SigBasedModel}

Assume, for simplicity, that $r = 0$, and let $(S_t)_{t\in [0,T]}$ denote the price process of a risky asset. We work under a risk-neutral measure, and model the discounted stock price $\Tilde{S}_t := S_t$ as 
\[
d\Tilde{S}_t = \Tilde{S}_t \sigma_t \, dB_t,
\]
where $\sigma_t$ is the volatility and $B$ is a standard Brownian motion. 

Let $f : \Lambda^p_T(V) \to \R$ be a continuous function and suppose that volatility can be expressed as
\[
\sigma_t = f\big((\bX_s)_{s \in [0,t]}\big),
\]
where the stochastic process $\bX \in \Lambda^p_T(V)$ is referred to as the \emph{primary process}. This process represents the underlying source of noise driving the volatility and is assumed to take values in the space of stopped, time-augmented weakly geometric $p$-rough paths.

Fix an integer $N \geq 1$ and define the space
\[
A_N := \operatorname{span} \left\{ \bX \mapsto \langle e_I, S(\bX)_T \rangle ;\ I \text{ multi-index in } \R \oplus V,\, 0 \leq |I| \leq N \right\},
\]
where $S(\bX)_t := S(_{[\mathsf{t}]}\bX)_{0,t}$ denotes the signature of the stopped path $_{[\mathsf{t}]}\bX$. 

For a given $N\geq 1$, our goal is to find the coefficients $\ell = \{\ell_I; \, 0 \leq |I| \leq N\}$ that yield the best linear approximation $\langle \ell, S(\bX)^{\leq N}_t \rangle$ of the volatility $\sigma_t$, where
\[
\langle \ell, S(\bX)^{\leq N}_t \rangle = \sum_{|I| \leq N} \ell_I \langle e_I, S(\bX)^{\leq N}_t \rangle
\]
with $S(\bX)^{\leq N}_t \in T^N(\R \oplus V)$. This leads to the following signature-based stochastic volatility model:
\begin{align} 
    d\Tilde{S}_t(\ell) & = \Tilde{S}_t(\ell)\, \sigma_t(\ell)\, dB_t, \label{signatmodel1} \\
    \sigma_t(\ell) & = \sum_{|I| \leq N} \ell_I \langle e_I, S(\bX)^{\leq N}_t \rangle. \nonumber
\end{align}
Since $\bX$ is a stochastic process, the signature $S(\bX)^{\leq N}$ is itself a stochastic process taking values in the truncated tensor algebra $T^N(\R \oplus V)$.

\begin{ex}
    Assume that the dynamics of $(S_t)_{t\geq 0}$ are given by
    \[
        dS_t = r S_t\, dt + \sigma_t S_t\, d(\rho W_t + \sqrt{1 - \rho^2}\, B_t), \nonumber \\
    \]
    where $W$ and $B$ are independent standard Brownian motions. As the primary process $(X_t)_{t\geq 0}$ with dynamics
\begin{equation} \label{XprocessExample}
    dX_t = \kappa(\theta - X_t)\, dt + \nu \sqrt{X_t}\, dW_t
\end{equation}
is a continuous semimartingale, it admits a canonical weakly geometric rough path lift (see Remark~\ref{semimartingale_lift}). By Proposition~\ref{augmentXtohatX}, the time-augmented process also admits a weakly geometric $p$-rough path lift $\bX \in WG\Omega^p_T(\R \oplus \R)$, where time corresponds to coordinate $0$ and $X_t$ to coordinate $1$. Note that the Feller condition $2\kappa \theta \geq \nu^2$ ensures positivity of $X_t$, but is not required for the existence of the rough path lift.

\noindent We model the volatility as a linear function of the truncated time-augmented signature up to level $2$ of the primary process $X$:  
    \[
    \sigma_t(\ell):= \sum_{|I| \leq 2} \ell_I \langle e_I, S(\bX)^{\leq 2}_t \rangle.
    \]
    The truncated signature has the form
    \[
    S(\bX)_t^{\leq 2} = \left(1,\; (t, X_t),\;
    \begin{pmatrix}
        \int_0^t s\, ds & \int_0^t s\, dX_s \\
        \int_0^t X_s\, ds & \int_0^t X_s\, dX_s
    \end{pmatrix}
    \right).
    \]
    If we let
    \[
    \ell = \left(\ell_{\emptyset}, \ell_0, \ell_1, \ell_{00}, \ell_{01}, \ell_{10}, \ell_{11} \right),
    \]
    we then have
    \[
    \sigma_t(\ell) = \ell_{\emptyset} + \ell_0 t + \ell_1 X_t + \ell_{00} \frac{t^2}{2} + \ell_{01} \int_0^t s\, dX_s + \ell_{10} \int_0^t X_s\, ds + \ell_{11} \int_0^t X_s\, dX_s.
    \]
    Differentiating with respect to $t$ and using (\ref{XprocessExample}) yields
    \begin{align*}
        d\sigma_t(\ell) &= \left( \ell_0 + \ell_{00} t + \ell_{10} X_t + (\ell_1 + \ell_{01} t + \ell_{11} X_t)\, \kappa(\theta - X_t) \right)\, dt \\
        &\quad + \nu \sqrt{X_t}\, \left(\ell_1 + \ell_{01} t + \ell_{11} X_t\right)\, dW_t.
    \end{align*}
    This expression shows how a linear signature model recovers a volatility process consistent with the Heston variance dynamics.
\end{ex}

    Continuing with the example, the third level of the time-augmented signature, $S(\bX)_t^3$, consists of iterated integrals over $0 < u_1 < u_2 < u_3 < t$, indexed by $I = (i_1, i_2, i_3) \in \{0,1\}^3$. A few representative terms are:
    \[
    \begin{aligned}
    \langle e_{000}, S(\bX)_t \rangle &= \int_0^t \int_0^{u_2} \int_0^{u_1} du_0\, du_1\, du_2 = \frac{t^3}{6}, \\
    \langle e_{001}, S(\bX)_t \rangle &= \int_0^t \int_0^{u_2} \int_0^{u_1} du_0\, du_1\, dX_{u_2}, \\
    \langle e_{011}, S(\bX)_t \rangle &= \int_0^t \int_0^{u_2} \int_0^{u_1} du_0\, dX_{u_1}\, dX_{u_2}, \\
    \langle e_{111}, S(\bX)_t \rangle &= \int_0^t \int_0^{u_2} \int_0^{u_1} dX_{u_0}\, dX_{u_1}\, dX_{u_2}.
    \end{aligned}
    \]
    In total, the third level of the time-augmented signature includes $2^3 = 8$ terms:
    \[
    \left(
    \begin{array}{cc}
        \langle e_{000}, S(\bX)_t \rangle & \langle e_{001}, S(\bX)_t \rangle \\[1ex]
        \langle e_{010}, S(\bX)_t \rangle & \langle e_{011}, S(\bX)_t \rangle \\[1ex]
        \langle e_{100}, S(\bX)_t \rangle & \langle e_{101}, S(\bX)_t \rangle \\[1ex]
        \langle e_{110}, S(\bX)_t \rangle & \langle e_{111}, S(\bX)_t \rangle
    \end{array}
    \right)
    =
    \left(
    \begin{array}{ll}
        \int_0^t \int_0^{u_2} \int_0^{u_1} du_0\, du_1\, du_2
        & \int_0^t \int_0^{u_2} \int_0^{u_1} du_0\, du_1\, dX_{u_2} \\[1ex]
        \int_0^t \int_0^{u_2} \int_0^{u_1} du_0\, dX_{u_1}\, du_2
        & \int_0^t \int_0^{u_2} \int_0^{u_1} du_0\, dX_{u_1}\, dX_{u_2} \\[1ex]
        \int_0^t \int_0^{u_2} \int_0^{u_1} dX_{u_0}\, du_1\, du_2
        & \int_0^t \int_0^{u_2} \int_0^{u_1} dX_{u_0}\, du_1\, dX_{u_2} \\[1ex]
        \int_0^t \int_0^{u_2} \int_0^{u_1} dX_{u_0}\, dX_{u_1}\, du_2
        & \int_0^t \int_0^{u_2} \int_0^{u_1} dX_{u_0}\, dX_{u_1}\, dX_{u_2}
    \end{array}
    \right)
    \]

In Sections \ref{numerical_Heston} and \ref{numerical_rBergomi} we will work with the truncated signature $S(\bX)^{\leq 3}_t$.

\subsection{The Signature Approximation to Volatility}\label{signatureappoximation}

If $V$ has dimension $d$, let $d_N := \sum_{k=0}^N d^k$ denote the dimension of the truncated tensor algebra $T^N(V)$. The following notation allows us to represent elements of $T^N(V)$ as vectors in $\R^{d_N}$. Let
\[
\mathcal{L} : \{ I;\ |I| \leq N \}\ \to\ \{1,\dots,d_N\}
\]
be a \emph{labeling} function, that is, a bijection that assigns a unique index to each multi-index $I$ of length at most $N$. For any $\ell = \sum_{|I| \leq N} \ell_I e_I \in T^N(V)$, the map
\[
\begin{array}{rccl}
\mathbf{vec} : & T^N(V) & \to & \R^{d_N} \\[1ex]
& \ell & \mapsto & \big( \ell_{\mathcal{L}^{-1}(1)}, \dots, \ell_{\mathcal{L}^{-1}(d_N)} \big)
\end{array}
\]
\emph{flattens} the elements of $T^N(V)$, making it possible to identify tensors with vectors in $\R^{d_N}$. This makes the use of signature data in numerical algorithms quite convenient. We now apply this to the truncated signature $\Tilde{S}_t(\ell)$.

\medskip 
Assume that the primary process $\bX \in \Lambda^p_T(V)$ is obtained as the rough path lift of a stochastic process $(t, X_t)$, where $X_t$ solves an SDE driven by a Brownian motion $W_t$. In particular, $\bX$ is measurable with respect to $W$ and adapted to its natural filtration.

As above, let $f : \Lambda^p_T(V) \to \R$ be a continuous function and let $\sigma_t = f\big((\bX_s)_{s \in [0,t]}\big)$. Define the process $Z_t := \rho W_t + \sqrt{1 - \rho^2}\, B_t$, where $B$ is another Brownian motion, independent of $W$. Our model is then
\begin{equation} \label{signatmodel2}
\begin{aligned}
    d\Tilde{S}_t(\ell) & = \Tilde{S}_t(\ell)\, \sigma_t(\ell)\, dZ_t  \\ 
    \sigma_t(\ell) & \approx \sum_{|I| \leq N} \ell_I \langle e_I, S(\bX)^{\leq N}_t \rangle,
\end{aligned}
\end{equation}
where \eqref{signatmodel2} is a well-defined Itô integral, since $\sigma_t(\ell)$ is a predictable process adapted to the filtration of $W$, and $Z$ is a Brownian motion correlated with $W$. Recall that 
$S(\bX)_t = S(_{[\mathsf{t}]}\bX)_{0,t}$ is the signature of the stopped path $_{[\mathsf{t}]}\bX$. The solution of \eqref{signatmodel2} can be expressed in terms of the signature as follows.

\begin{prop}
Let $\Tilde{S}_t(\ell)$ be the discounted price process defined by \eqref{signatmodel2}, and assume that $\bX \in WG\hat{\Omega}^p_T(V)$ is the time-augmented weakly geometric $p$-rough path lift of a stochastic process $(t, X_t)$ adapted to a Brownian motion $W$. Let $Z_t = \rho W_t + \sqrt{1 - \rho^2}\, B_t$, where $B$ is a Brownian motion independent of $W$. Then, for a given $N\geq 1$, the discounted price process $\Tilde{S}_t(\ell)$ admits the representation
\[
\Tilde{S}_t(\ell) = S_0 \exp \left( \ell^T Q(t) \ell + \ell^T \int_0^t \mathbf{vec}(S(\bX)^{\leq N}_s)\, dZ_s \right),
\]
where $\ell^T$ denotes the transpose of $\ell$, and $Q(t)$ is the symmetric matrix defined by
\[
Q(t)_{\mathcal{L}(I), \mathcal{L}(J)} := -\frac{1}{2} \left\langle (e_I \shuffle e_J) \otimes e_0,\; S(\bX)^{\leq 2N+1}_t \right\rangle.
\]
\end{prop}
   
\begin{proof}
    Rewrite (\ref{signatmodel2}) as
    \[
    \frac{d\Tilde{S}_t(\ell)}{\Tilde{S}_t(\ell)} = \sigma_t(\ell)\, dZ_t.
    \]
    Applying It\^o's formula, we have
    \begin{align*}
    d\log(\Tilde{S}_t(\ell)) & = \frac{d\Tilde{S}_t(\ell)}{\Tilde{S}_t(\ell)} + \frac{1}{2} \left(\frac{-1}{\Tilde{S}_t(\ell)^2} \right) \sigma_t(\ell)^2 \Tilde{S}_t(\ell)^2 dt \\
    & = - \frac{1}{2} \sigma_t(\ell)^2 dt + \sigma_t(\ell)\, dZ_t.
    \end{align*}
    Integrating and rearranging we get
    \begin{equation}\label{TildeSproof}
    \Tilde{S}_t(\ell) = S_0 \exp\left( -\frac{1}{2}\int_0^t \sigma_s(\ell)^2 ds + \int_0^t \sigma_s(\ell)dZ_s \right) .
    \end{equation}
    With the notation introduced above, we can express the linear approximation of order $N$ to the volatility as
    \[
    \sigma_t(\ell) = \ell^T \mathbf{vec}\left(S(\bX)^{\leq N}_t \right).
    \]
    It follows that
    \[
    \sigma_t(\ell)^2 = \left( \ell^T \mathbf{vec}\left(S(\bX)^{\leq N}_t \right) \right)  \left( \ell^T \mathbf{vec}\left(S(\bX)^{\leq N}_t \right) \right)^T =\ \ell^T \mathbf{vec}\left (S(\bX)^{\leq N}_t \right) \mathbf{vec}\left( S(\bX)^{\leq N}_t)\right)^T \ell. 
    \]
    Define $\Tilde{Q}(t):= \mathbf{vec}\left(S(\bX)^{\leq N}_t \right) \mathbf{vec}\left(S(\bX)^{\leq N}_t\right)^T$, which is a symmetric positive semi-definite matrix representing the outer product of the truncated signature vector with itself. We then have
    \[
    \sigma_t(\ell)^2 = \ell^T \Tilde{Q}(t)\, \ell.
    \] 
    Note that the elements of $\Tilde{Q}(t)$ are
    \[
    \Tilde{Q}(t)_{\mathcal{L}(I) \mathcal{L}(J)} = \langle e_I, S(\bX)^{\leq N}_t\rangle \langle e_J, S(\bX)^{\leq N}_t\rangle = \langle e_I \shuffle e_J, S(\bX)^{\leq 2N}_t \rangle.
    \]
    It follows that  
    \begin{align}
        \int_0^t \Tilde{Q}(s)_{\mathcal{L}(I) \mathcal{L}(J)} ds
        & = \int_0^t \langle e_I \shuffle e_J, S(\bX)_{s}^{\leq 2N} \rangle ds \nonumber \\
        & = \langle (e_I \shuffle e_J) \otimes e_0, S(\bX)_t^{\leq 2N+1} \rangle. \label{QtildeShufffle}
    \end{align}
    If we define the matrix $Q(t)$ by 
    \begin{equation}
    Q(t)_{\mathcal{L}(I) \mathcal{L}(J)} := -\frac{1}{2} \langle (e_I \shuffle e_J) \otimes e_0, S(\bX)_t^{\leq 2N+1} \rangle, \label{QShuffle} 
    \end{equation}
    then the first term in the exponential of (\ref{TildeSproof}) can be written as 
    \begin{equation*}
        -\frac{1}{2}\int_0^t \sigma_s(\ell)^2 ds = \ell^T \left(-\frac{1}{2} \int_0^t \Tilde{Q}(s)\, ds \right)\ell = \ell^T Q(t) \ell.
    \end{equation*}
    The second term in the exponential is 
    \[
    \int_0^t \sigma_s(\ell) dZ_s = \ell^T \int_0^t \mathbf{vec}\left( S(\bX)^{\leq N}_s \right) dZ_s,
    \]
    which concludes the proof.
\end{proof}

Since $\Tilde{Q}(t)$ is positive semi-definite by construction, it follows that $Q(t)$ is negative semi-definite. That is, for every $\ell \in \R^{d_N}$,
\[
\ell^T Q(t) \ell \leq 0.
\]
As the Cholesky decomposition applies to positive semi-definite matrices, $-Q(t)$ admits a Cholesky decomposition $U(t)^T U(t)$. We can therefore write
\[
\ell^T Q(t) \ell = -||U(t) \ell||^2_{2}.
\]
Note that it is cheaper to compute $||U(t) \ell||^2_{2}$ than $\ell^T Q(t) \ell$. 

\begin{remark}
In the case of a one-dimensional primary process, we see from (\ref{QShuffle}) that $Q(t)$ depends on the signature $S(\bX)_t$ up to level $2N+1$.
\end{remark}

If the signature is truncated at level $N=3$, then $\Tilde{Q}(t)$ is a $15\times 15$-matrix, whose elements are of the type $\langle e_I, S(\bX)^{\leq 3}_t\rangle$, where the $15$ basis elements are
\[
\{1, e_0, e_1, e_{00}, e_{01}, e_{10}, e_{11}, e_{000}, e_{001}, e_{010}, e_{011}, e_{100}, e_{101}, e_{110}, e_{111}\}.
\]
With the above ordering, $\mathcal{L}(e_{01}) = 5$ and $\mathcal{L}(e_{111}) = 15$. The $(15,5)$-entry of matrix $\Tilde{Q}(t)$ is therefore
\[
\Tilde{Q}(t)_{15,5} = \langle e_{111}, S(\bX)^{\leq 3}_t\rangle \langle e_{01}, S(\bX)^{\leq 3}_t\rangle. 
\]
By (\ref{QtildeShufffle}),  
\[
\int_0^t \Tilde{Q}(s)_{15,5}\, ds = \langle (e_{111} \shuffle e_{01}) \otimes e_0, S(\bX)_t^{\leq 7} \rangle, 
\]
where
\begin{align*}
    (e_{111} \shuffle e_{01}) \otimes e_0 & = (e_{11101} + 2e_{11011} + 3e_{10111} + 4e_{01111})\otimes e_0 \\
    & = e_{111010} + 2e_{110110} + 3e_{101110} + 4e_{011110},
\end{align*}
which includes basis elements of tensor spaces of higher dimensions. In particular, to compute the matrix $Q(t)$, we see from (\ref{QShuffle}) that its entries $\langle (e_I \shuffle e_J) \otimes e_0, S(\bX)_t^{\leq 2N+1} \rangle$ live in a space of dimension $2^{(2N+1)+1}-1$. 

In our particular case, $2N+1 = 7$, so $Q(t)$ depends on signature entries in $T^{\leq 7}(\R \oplus V)$, a space of dimension $2^8 - 1 = 255$. To populate a $15\times 15$ matrix we need to fetch its elements from the entries of a $255\times 255$ matrix.   

\subsection{Calibration} \label{calibration_signat}

Consider the following model of a discounted asset price, parameterized by $\theta$:
\[
d\Tilde{S}_t^{\theta} = \Tilde{S}_t^{\theta}\, \sigma_t^{\theta}\, dZ_t,
\]
where $\sigma_t^{\theta}$ is the volatility process associated with the parameter $\theta$, and $Z_t$ is as in (\ref{signatmodel2}). For a given $\theta$, we can compute the price of a European call option with strike $K$ and maturity $T$ as
\[
C(K, T, \theta) = e^{-rT} \mathbb{E}[(S_T^{\theta} - K)_+] = \mathbb{E}[(\Tilde{S}_T^{\theta} - e^{-rT} K)_+].
\]

Let $\{ C^{\textsf{mkt}}(K_i, T_i) \}_{i=1}^N$ denote the observed market prices for varying strikes and maturities. If there exists a parameter $\theta^*$ such that the model perfectly describes the real dynamics of the asset, then
\[
C^{\textsf{mkt}}(K_i, T_i) = C(K_i, T_i, \theta^*) \quad \text{for all } i.
\]

While this exact match is unlikely in practice, our goal is to find a parameter configuration $\theta$ that minimizes the discrepancy between the model and market prices. We thus consider the least squares loss function
\[
    L(\theta) = \sum_{i=1}^N \gamma_i \left( C^{\textsf{mkt}}(K_i, T_i) - C(K_i, T_i, \theta) \right)^2,
\]
where $\gamma_i > 0$ are user-specified weights.

In the case of signature-based models, the role of the parameter $\theta$ is played by a vector $\ell \in \R^{d_N}$, where $d_N$ is the dimension of the truncated signature space. The corresponding loss function becomes
\begin{equation}\label{LSlossfunction}
L(\ell) = \sum_{i=1}^N \gamma_i \left( C^{\textsf{mkt}}(K_i, T_i) - C(K_i, T_i, \ell) \right)^2.
\end{equation}

Recall that the value of the signature-driven price process $\Tilde{S}_t(\ell)$ at maturity $t = T$ is given by
\begin{equation}\label{priceTildeSexpression}
    \Tilde{S}_T(\ell)(\omega) = S_0 \exp \left( -\| U(T)(\omega)\, \ell \|_2^2 + \ell^T \int_0^T \mathbf{vec}\left( S(\bX)_t^{\leq N}(\omega)\right)\, dZ_t \right),
\end{equation}
where $U(T)(\omega)$ is the Cholesky factor associated with $-Q(T)$ on sample path $\omega$, and $S(\bX)_t^{\leq N}$ is the truncated signature of the primary process $\bX$ up to level $N$.

\subsection{The Algorithm} \label{algorithm}

We start by simulating the discounted stock prices $\Tilde{S}_T$ at maturities $T \in \{0.1,\ 0.6,\ 1.1,\ 1.6\}$. For each $T$, we need to compute the matrix $U(T)$ and the stochastic integrals $\int_0^T \mathbf{vec}(S(\bX)^{\leq N}_t)\,dZ_t$ that appear in~\eqref{priceTildeSexpression}. Let $n_{\text{MC}}$ be the number of Monte Carlo samples; for each maturity, we compute the call prices corresponding to strikes $K \in \{90,\ 95,\ 100,\ 105,\ 110\}$, which yields the following 20 values:
\[
C(K_i, T_i, \ell) \approx \frac{1}{n_{\text{MC}}} \sum_{j=1}^{n_{\text{MC}}} \left( \Tilde{S}_{T_i}(\ell)(\omega_j) - e^{-r T_i} K_i \right)_{+},
\]
$i=1,\dots,20$, where each $\omega_j$ denotes a sample path. 

For calibration, we use as ground truth the synthetic market prices $\{C^{\textsf{mkt}}(K_i, T_i)\}_{i=1}^{20}$ generated under the assumption that the market follows either Heston (Section~\ref{numerical_Heston}) or rough Bergomi dynamics (Section~\ref{numerical_rBergomi}).  

The signature approach seeks to minimize the discrepancy between market option prices and signature-generated prices, which is achieved by minimizing 
\[
L(\ell) = \sum_{i=1}^{20} \gamma_i \left( C^{\textsf{mkt}}(K_i, T_i) - C(K_i, T_i, \ell) \right)^2,
\]
where the weights $\gamma_i$ are proportional to the inverse Vega of each option. 

Once the optimal coefficient vector $\ell^*$ is obtained, we can generate three sets of option prices. We describe here the Heston case:\footnote{The superscript $\textsf{SIG}$ refers to results from the signature-based approach. In the case of the Heston approximation of \cite{ADSV15}, we use the superscript $\textsf{ASV}$ (from the authors’ surnames). In Section~\ref{numerical_rBergomi}, the prices obtained with the analytical approximation described in Section~\ref{short_time_calib_rb} will be denoted by superscript $\textsf{VIX}$.}:
\begin{itemize}
    \item $\{C^{\textsf{mkt}}(K_i, T_i)\}$, the synthetic "market" prices,
    \item $\{C(K_i, T_i, \ell^*)\}$, the signature model prices,
    \item $\{C^{ASV}(K_i, T_i)\}$, the prices using the second-order approximation in \cite{ADSV15}.
\end{itemize}
\noindent Using these prices, we compute the three implied volatility surfaces from the Black-Scholes formula:
\begin{itemize}
    \item $\{\text{IV}^{\textsf{mkt}}(K_i, T_i)\}$, from the "market" prices,
    \item $\{\text{IV}^{SIG}(K_i, T_i, \ell^*)\}$, from the signature model prices,
    \item $\{\text{IV}^{ASV}(K_i, T_i)\}$, from the second-order approximation prices.
\end{itemize}
\noindent These surfaces are compared in Section~\ref{numerical_Heston}. We now describe the algorithm in detail.

\begin{enumerate}
    \item \textbf{Simulate sample paths.} Generate $n_{\text{MC}}$ Monte Carlo paths for the Brownian motions $W$ and $B$ using Gaussian increments. Construct $Z = \rho W + \sqrt{1 - \rho^2} B$, and simulate the process $X$ using an Euler scheme. Construct the augmented process $\bX$, and for each path:
    \begin{itemize}
        \item compute the truncated signature $S(\bX)_T^{\leq 2N+1}$,
        \item evaluate the stochastic integral $\int_0^T \mathbf{vec}(S(\bX)^{\leq N}_t)\, dZ_t$.
    \end{itemize}
    
    \item \textbf{Assemble the matrix $Q(T)$.} For each sample path, compute the symmetric matrix
    \[
    Q(T)_{\mathcal{L}(I), \mathcal{L}(J)} = -\frac{1}{2} \left\langle (e_I \shuffle e_J) \otimes e_0,\; S(\bX)_T^{\leq 2N+1} \right\rangle,
    \]
    and perform a Cholesky decomposition of $-Q(T)$ to obtain $U(T)$. 

    \item \textbf{Optimize the loss.} Initialize $\ell \in \R^{d_N}$ and iterate the following steps until convergence:
    \begin{enumerate}
        \item For each path $\omega_j$, evaluate
        \[
        \Tilde{S}_T(\ell)(\omega_j) = S_0 \exp\left( -\|U(T)(\omega_j)\, \ell\|_2^2 + \ell^T \int_0^T \mathbf{vec}(S(\bX)^{\leq N}_t(\omega_j))\, dZ_t \right).
        \]
        \item Compute $C(K_i, T_i, \ell)$ as the Monte Carlo average over $\omega_j$.
        \item Evaluate $L(\ell)$ and update $\ell$ using a numerical optimizer.
    \end{enumerate}
\end{enumerate}

\noindent Note that the signatures are computed once (\emph{offline}) and reused when updating $\ell$, making calibration significantly faster.

Before presenting the results, we highlight an important structural property of signatures that will serve as a diagnostic for numerical approximation quality.

\begin{prop}[Factorial Decay]\label{factorialdecay}
Let $X : [0,T] \to \R^d$ be a path of finite $p$-variation for some $p \geq 1$, and let $\bX \in WG\hat{\Omega}^p_T(\R^d)$ denote its time-augmented weakly geometric rough path lift. Then for all $k \geq 1$, the $k$-th level satisfies
\[
\| \bX^k_{s,t} \| \leq \frac{C(X)^k}{k!},
\]
for some constant $C(X) > 0$ depending on $X$, uniformly over all $(s,t) \in \Delta_T$, and for any tensor norm on $(\R^d)^{\otimes k}$.
\end{prop}

This factorial decay follows from the multiplicative (group-like) structure and the control provided by the $p$-variation norm. In practice, it serves as a valuable check: the magnitudes of the iterated integrals should decay rapidly with $k$, and deviations from this pattern can signal numerical instability or truncation issues. See \cite{lyons98} and \cite[Thm.~10.35]{frizvictoir10} for proofs and generalizations.

For paths of bounded variation, a stronger estimate holds:
\[
\| \bX^k_{s,t} \| \leq \frac{1}{k!} \| \bX \|^k_{1\cdot \mathrm{var}},
\]
as noted in \cite{fermanian21}. This bound is exact for signatures of bounded variation paths. Although it does not hold for arbitrary $p$-rough paths, it remains relevant numerically, since signature approximations are typically based on interpolated (and hence $BV$) paths.

\noindent \textbf{Implementation details.} 
Results in the following sections correspond to $n_{\text{MC}} = 800{,}000$ Monte Carlo paths and signature truncation level $N = 3$. Brownian increments are generated via standard Gaussian sampling, and $X$ is simulated using an Euler discretization.

Signatures are computed using a vectorized version of Peter Foster's code\footnote{{\url{https://github.com/pafoster/path\_signatures\_introduction}}}, adapted for GPU acceleration. Optimization of $L(\ell)$ is done using SciPy's \texttt{minimize} function with the \texttt{L-BFGS-B} method (tolerance $10^{-8}$), and with box constraints on $\ell$ to accelerate convergence.

\cite{issa23} note that the choice of interpolation method typically has little impact and that simple linear interpolation is often sufficient to compute signature approximations. We did experiment with higher-order interpolation schemes (such as cubic splines), but the marginal gains in accuracy were negligible, so we kept it linear.

All computations were carried out on a consumer desktop with 128\,GB RAM and an NVIDIA RTX 3080 Ti GPU, without access to specialized computing clusters.

\section{Calibration with a Heston Primary Process}\label{numerical_Heston}

In this section, we compare the performance of the signature-based method introduced in Section~\ref{SigBasedModel} with the parametric approach presented in Section~\ref{ASV_IVolApproxim}. 

As the signature-based approach learns volatility directly from a \emph{primary} noise, we first tested this learning mechanism using an Ornstein–Uhlenbeck process. However, to compare fairly with the parametric calibration in \cite{ADSV15}, which is derived under Heston dynamics, we use a Heston variance primary process.

\subsection{Calibration Setup. The Uncorrelated Case.}

Recall that the \emph{market} model in \cite{ADSV15} is given by:
\begin{align*}
    dS_t &= r S_t\, dt + \sigma_t S_t\, d\big( \rho W_t + \sqrt{1 - \rho^2}\, B_t \big), \\
    d\sigma_t^2 &= \kappa(\theta - \sigma_t^2)\, dt + \nu \sqrt{\sigma_t^2}\, dW_t.
\end{align*}

Using the calibrated parameters from Table~\ref{tab:ASV-calibration}, we compute option prices and invert the Black-Scholes formula to obtain the implied volatility surface, which we denote by $\text{IV}^{\textsf{ASV}}$.

To calibrate the signature-based model, we consider 20 option prices $C(K_i, T_i)$ at maturities $\{0.1,\, 0.6,\, 1.1,\, 1.6\}$ and strikes $\{90,\, 95,\, 100,\, 105,\, 110\}$. The primary process $X$ follows a Heston variance SDE with parameters $X_0 = 0.1$, $\nu = 0.2$, $\kappa = 2$, and $\theta = 0.15$.

Volatility is modeled as:
\[
\sigma_t(\ell) \approx \langle \ell, S(\bX)_t^{\leq 3} \rangle,
\]
where $S(\bX)_t^{\leq 3}$ is the truncated signature of the time-augmented path $\hat{X}_t = (t, X_t)$. The loss function to minimize is:
\begin{equation}\label{lossfunction1}
    L(\ell) = \sum_{i=1}^{20} \gamma_i \left( C^{\textsf{mkt}}(K_i, T_i) - C(K_i, T_i, \ell) \right)^2,
\end{equation}
where $\gamma_i$ is set as the inverse Vega of the $i$-th option.

As mentioned above, the number of Monte Carlo samples is $n_{\text{MC}} = 800{,}000$. Since computing the matrix $Q(T)$ involves terms up to level $2N + 1 = 7$, this requires evaluating signatures in a $2^8 - 1 = 255$-dimensional space. That is, each $15 \times 15$ matrix $Q(T)$ depends on a corresponding $255 \times 255$ matrix, making this step computationally intensive.

\paragraph{Calibration Results.} The optimal coefficient vector minimizing $L(\ell)$ is:
\begin{align*}
    \ell^* = (& 0.201202133,\; 0.142660997,\; 1.08471290,\; -0.297312378,\; -0.0293435325,\; -0.0422317187, \\
           & 9.25090162 \times 10^{-4},\; 0.293103687,\; -0.0143435573,\; -0.0134285652,\; -1.64737083 \times 10^{-3}, \\
           & -2.89883092 \times 10^{-3},\; -5.72798006 \times 10^{-4},\; -1.93045420 \times 10^{-3},\; -1.84406803 \times 10^{-4}).
\end{align*}

\noindent Recall that
\[
\sigma_t(\ell) = \ell_{\emptyset} + \ell_0 t + \ell_1 X_t + \ell_{00} \frac{t^2}{2} + \ell_{01} \int_0^t s\, dX_s + \ell_{10} \int_0^t X_s\, ds + \ell_{11} \int_0^t X_s\, dX_s + \ell_{000} \frac{t^3}{6} + \cdots.
\]
The coefficient $\ell_1 \approx 1.085$ confirms that the model has learned a strong linear dependence on $X_t$, consistent with the Heston structure. Likewise, $\ell_{\emptyset} \approx 0.201$ is close to the initial volatility $\sigma_0=0.2$.

The minimum value of the loss function $L(\ell)$ obtained during calibration was $1.05 \times 10^{-4}$, indicating a good fit to market prices. From the simulated option prices we compute the implied volatilities, which we denote by $\text{IV}^{\textsf{SIG}}$. Figure~\ref{figComparisonIVols1} shows a comparison between the two surfaces. 

\begin{figure}[H]
    \centering
    \includegraphics[width=0.5\linewidth]{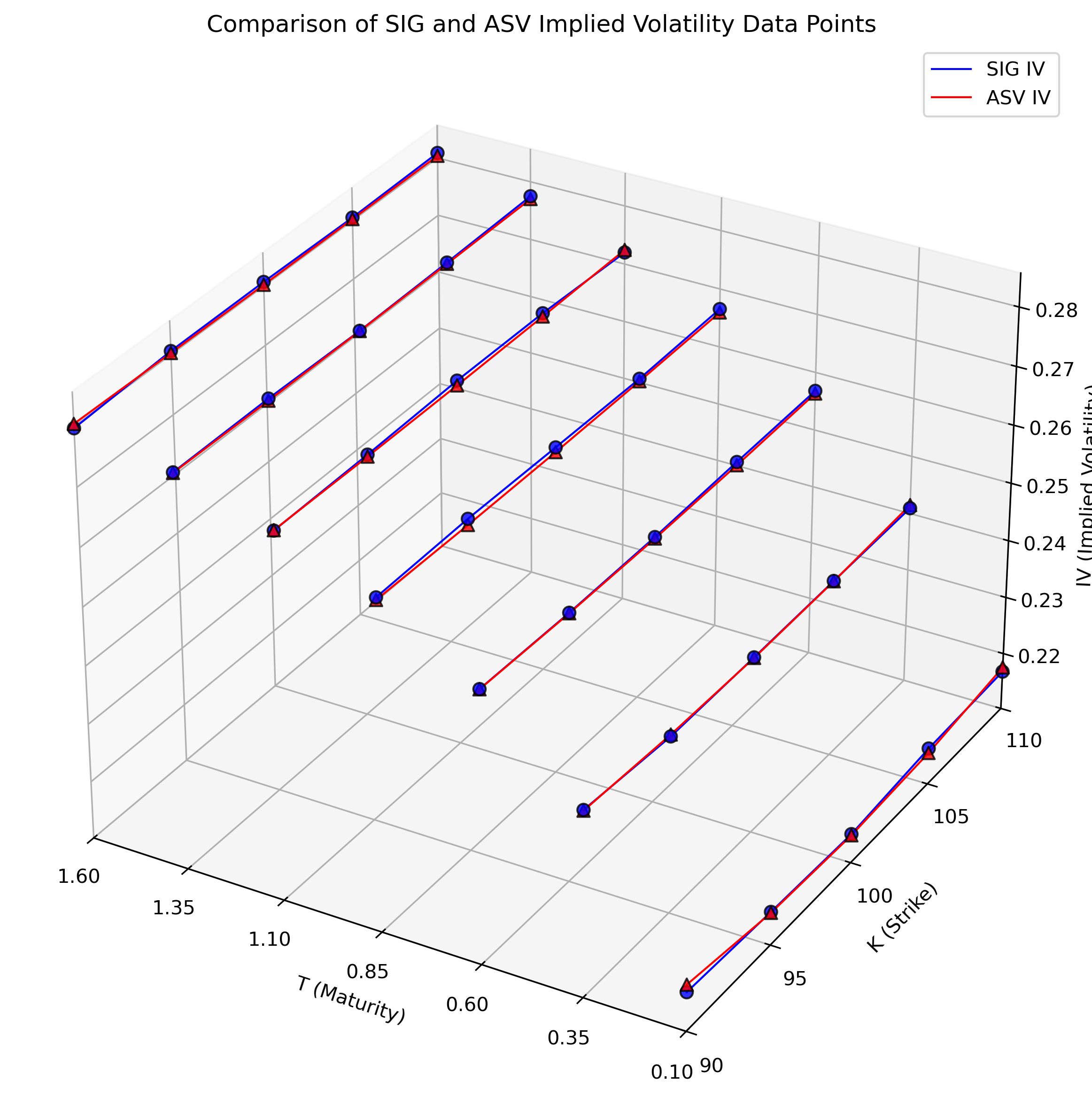}
    \caption{\small Heston implied volatility surfaces: signature-based (\textsf{SIG}) and analytical approximation (\textsf{ASV}).}
    \label{figComparisonIVols1}
\end{figure}

To evaluate the quality of the calibration methods, we computed, for each of the 20 option contracts $(K_i, T_i)$, the errors
\begin{align*}
    e^{\textsf{SIG}}_{i} \, & =\, \left| \text{IV}^{\textsf{SIG}}(K_i, T_i) - \text{IV}^{\textsf{mkt}}(K_i, T_i) \right| \\
    e^{\textsf{ASV}}_{i} \, & =\, \left| \text{IV}^{\text{ASV}}(K_i, T_i) - \text{IV}^{\textsf{mkt}}(K_i, T_i) \right|.
\end{align*}
A detailed breakdown is reported in Table~\ref{tab:Heston_errors}, where entries marked with $(*)$ correspond to the cases in which $e^{\textsf{SIG}} < e^{\textsf{ASV}}$.  

\begin{table}[H]
\centering
\scriptsize
\begin{tabular}{ccll}
\toprule
$T$ & $K$ & $e^{\textsf{ASV}}$ & $e^{\textsf{SIG}}$ \\
\midrule
0.1 &  90 & 0.00004 & 0.00127 \\
0.1 &  95 & 0.00002 & 0.00007 \\
0.1 & 100 & 0.00005 & 0.00009 \\
0.1 & 105 & 0.00003 & 0.00069 \\
0.1 & 110 & 0.00003 & 0.00078 \\
0.6 &  90 & 0.00010 & 0.00024 \\
0.6 &  95 & 0.00012 & 0.00021 \\
0.6 & 100 & 0.00012 & 0.00005 (*) \\
0.6 & 105 & 0.00012 & 0.00026 \\
0.6 & 110 & 0.00010 & 0.00019 \\
1.1 &  90 & 0.00011 & 0.00029 \\
1.1 &  95 & 0.00012 & 0.00008 (*) \\
1.1 & 100 & 0.00012 & 0.00055 \\
1.1 & 105 & 0.00012 & 0.00031 \\
1.1 & 110 & 0.00012 & 0.00069 \\
1.6 &  90 & 0.00011 & 0.00089 \\
1.6 &  95 & 0.00012 & 0.00014 \\
1.6 & 100 & 0.00012 & 0.00029 \\
1.6 & 105 & 0.00012 & 0.00008 (*) \\
1.6 & 110 & 0.00011 & 0.00031 \\
\bottomrule
\end{tabular}
\caption{Calibration errors for the uncorrelated Heston model ($\rho=0$).} 
\label{tab:Heston_errors}
\end{table}

Both methods exhibit a high level of calibration accuracy, with most errors lying in the range of $10^{-4}$ to $10^{-5}$, indicating that their performance is broadly comparable. While the analytical approach generally yields slightly smaller errors, there are several instances in which the signature-based method performs better, as highlighted in the table. This confirms that the signature-based approach is able to capture the structure of the implied volatility surface with a degree of accuracy comparable to model-based expansions, while maintaining its flexibility and model-agnostic nature.

\subsection{The Correlated Case}\label{sub:Heston_correlated}

We now consider the case when the asset price and the volatility process are correlated. Using the calibrated parameters from Table~\ref{tab:ASV-correlated}, we obtain the option prices and compute the implied volatility surface.

We let the \emph{primary} Heston variance process $X$ be initialized with parameters $X_0 = 0.25$, $\nu = 0.35$, $\kappa = 3.3$, $\theta = 0.15$ and $\rho = -0.5$, and we approximate the volatility by the truncated signature
\[
\sigma_t(\ell) \approx \langle \ell, S(\bX)_t^{\leq 3} \rangle.
\]
The loss function remains as in~\eqref{lossfunction1}. The optimal coefficient vector is
\begin{align*}
    \ell^* = (& -0.195158212,\; -0.250867130,\; -0.125195785,\; 0.606113847,\; -0.303740047, \\ 
           & 0.347580926,\; 0.136816382,\; -0.664746087,\; 0.563172308,\; 0.033241841, \\
           & 0.029376982,\; 0.019240593,\; -0.065104522,\; 3.67\!\times\! 10^{-5},\; -8.94\!\times\! 10^{-3}).
\end{align*}

Whereas in the uncorrelated setting, the optimizer happened to find a parameterization close to the true process, in the correlated case it settled on a different (but still effective) minimizer. The minimum value of the loss function $L(\ell)$ was $1.46 \times 10^{-3}$, slightly less precise than in the uncorrelated case. Figure~\ref{figComparisonIVols2} provides a visual comparison of the two calibrated surfaces.

\begin{figure}[H]
    \centering
    \includegraphics[width=0.5\linewidth]{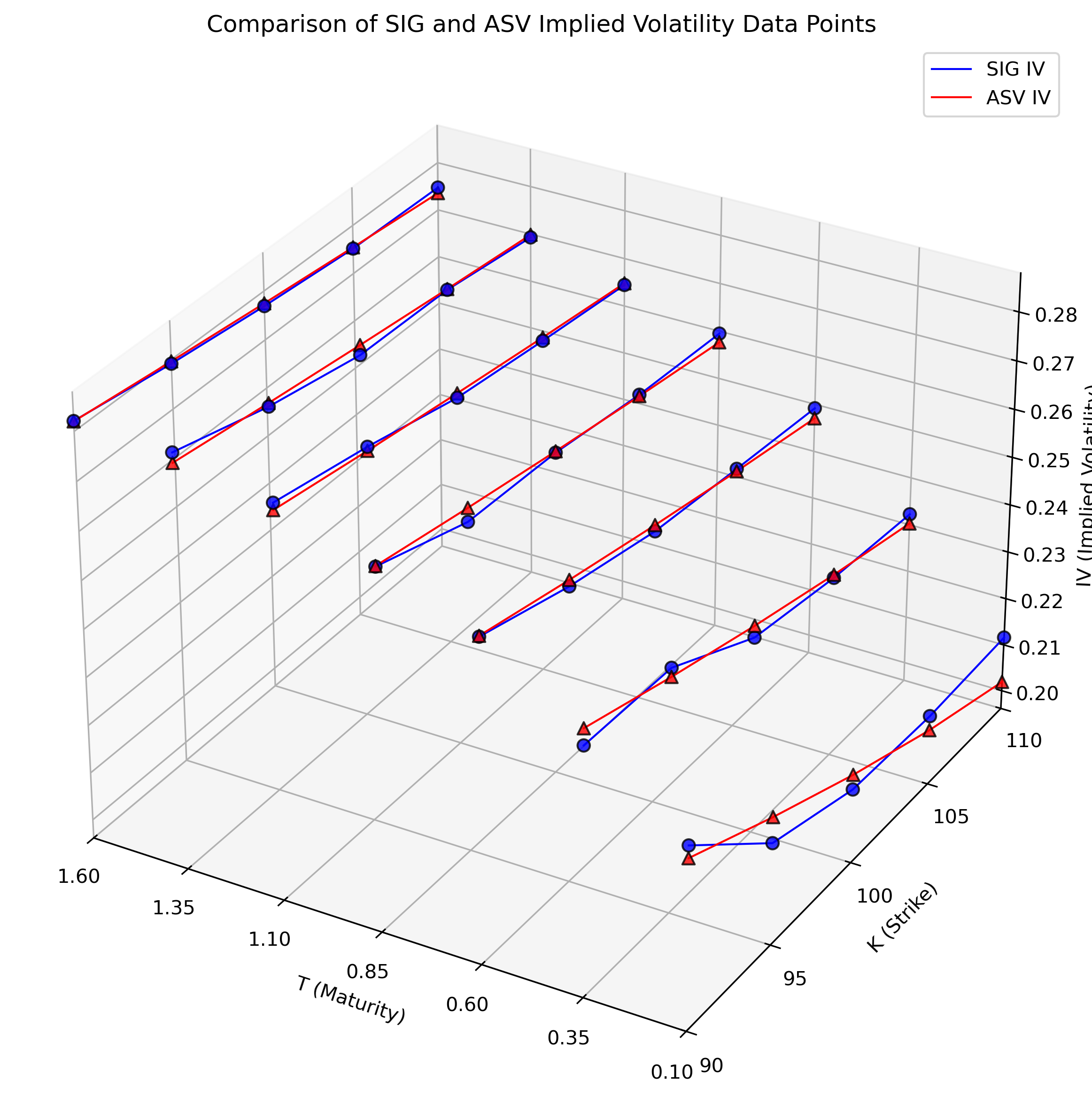}
    \caption{Implied volatility surfaces for the Heston model with correlation.}
    \label{figComparisonIVols2}
\end{figure}

To assess calibration accuracy, we compute the absolute errors $e^{\textsf{SIG}}_i$ and $e^{\textsf{ASV}}_i$ for each of the 20 contracts $(K_i, T_i)$. A detailed breakdown is provided in Table~\ref{tab:Heston_errors_correl}, where entries marked with $(*)$ correspond to the cases in which $e^{\textsf{SIG}} < e^{\textsf{ASV}}$.

\begin{table}[H]
\centering
\scriptsize
\begin{tabular}{ccll}
\toprule
$T$ & $K$ & $e^{\textsf{ASV}}$ & $e^{\textsf{SIG}}$ \\
\midrule
0.1 &  90 & 0.00046 & 0.00261 \\
0.1 &  95 & 0.00023 & 0.00535 \\
0.1 & 100 & 0.00002 & 0.00307 \\
0.1 & 105 & 0.00013 & 0.00304 \\
0.1 & 110 & 0.00018 & 0.00972 \\
0.6 &  90 & 0.00024 & 0.00018 (*) \\
0.6 &  95 & 0.00014 & 0.00130 \\
0.6 & 100 & 0.00005 & 0.00127 \\
0.6 & 105 & 0.00002 & 0.00051 \\
0.6 & 110 & 0.00008 & 0.00215 \\
1.1 &  90 & 0.00012 & 0.00146 \\
1.1 &  95 & 0.00006 & 0.00081 \\
1.1 & 100 & 0.00001 & 0.00091 \\
1.1 & 105 & 0.00004 & 0.00069 \\
1.1 & 110 & 0.00008 & 0.00036 \\
1.6 &  90 & 0.00005 & 0.00007 \\
1.6 &  95 & 0.00001 & 0.00047 \\
1.6 & 100 & 0.00003 & 0.00056 \\
1.6 & 105 & 0.00006 & 0.00027 \\
1.6 & 110 & 0.00010 & 0.00118 \\
\bottomrule
\end{tabular}
\caption{Calibration errors for the Heston case ($\rho=-0.5$).} 
\label{tab:Heston_errors_correl}
\end{table}

Both methods continue to exhibit a comparable level of accuracy, although the overall precision is slightly reduced in the presence of correlation, with most errors lying in the range of $10^{-4}$ to $10^{-3}$. The analytical approach generally achieves smaller errors, with only a limited number of instances in which the signature-based method performs better. One possible explanation for this behavior is that the effect of negative correlation is encoded in higher-order interactions, which are only partially captured at the chosen truncation level of the signature (see \cite{issa23}). Increasing the truncation level to $N=4$ yields only marginal improvements, suggesting that substantially higher-order terms may be required to fully capture the dependence structure, although at a significantly higher computational cost.

\section{Calibration with a Rough Bergomi Primary Process} \label{numerical_rBergomi}

The signature-based method makes no structural assumptions about the market volatility, offering greater flexibility and robustness. To demonstrate this flexibility, in this section we assume that \emph{the market is rough Bergomi} and we use a fractional Brownian motion as primary process.  

The asymptotic method in Section~\ref{ASV_IVolApproxim} is not suitable anymore because the information in the long-term maturities is not relevant for the rough Bergomi case. Instead, we rely on the algorithm described in Section~\ref{short_time_calib_rb}, which exploits the short-maturity behavior of European options.

We first compute the \emph{market} option prices and the option prices obtained with the calibrated parameters from Table~\ref{tab:rBergomi-tab}. For the numerical techniques, see \cite{BLPakkanen2017} and \cite{turbochargingRB2018}.\footnote{The code is available at \url{https://github.com/ryanmccrickerd/rough_bergomi}} We then invert both set of prices to obtain the implied volatility surfaces, which we denote respectively by $\text{IV}^{\textsf{mkt}}$ and $\text{IV}^{\textsf{VIX}}$. 

We now compare the parametric calibration $\text{IV}^{\textsf{VIX}}$ with the one obtained via the signature method. Let $Z_t = \rho W_t + \sqrt{1 - \rho^2}\, B_t$, where $B$ is a Brownian motion independent of $W$. The model is: 
\begin{equation*}
\begin{aligned}
    d\Tilde{S}_t(\ell) & = \Tilde{S}_t(\ell)\, \sigma_t(\ell)\, dZ_t  \\
    \sigma_t(\ell) & \approx \langle \ell, S(\bX)_t^{\leq 3} \rangle,
\end{aligned}
\end{equation*}
with primary process: 
\begin{equation}\label{primary_rBergomi}
X_t = \sqrt{2H}\int_0^t (t-s)^{H-\frac{1}{2}}\, dW_s.
\end{equation}
Unlike the case of continuous semimartingales, where the Stratonovich integral naturally provides the rough path structure, the rough path lift of fractional Brownian motion must be constructed using techniques from Gaussian process theory. (See~\cite{FrizHairer24}, Chapter 10). In particular, \cite{CoutinQian02} show that fractional Brownian motion admits a canonical geometric rough path lift for any Hurst parameter $H > 1/4$.

\begin{remark}\label{remark_Hroughness}
    In Section~\ref{short_time_calib_rb}, the \emph{market} values are generated from a rough Bergomi model with $H=0.1$, while in the calibration below we use a primary process~(\ref{primary_rBergomi}) with $H=0.2$. Both values lie below the theoretical threshold $H>1/4$. This does not pose a problem in practice: on the one hand, we work with the \emph{time-augmented} path $\hat{X}_t=(t,X_t)$, where the bounded variation time component provides additional structure; on the other hand, signatures are computed from discrete samples of the path, which are interpolated linearly (see Section~\ref{algorithm}), and such piecewise linear approximations always admit a rough path lift. 
\end{remark}

With these considerations in mind, we first calibrate the model using~(\ref{primary_rBergomi}) as the primary process, with parameters $H=0.2$ and $\rho=-0.6$. Using a fractional Brownian motion directly as the primary process yields accurate results, but the calibration is computationally demanding, as the model must implicitly learn to enforce positivity of the variance.  

Next, we consider the geometric transformation
\[
X_t = \exp\left(\sqrt{2H}\int_0^t (t-s)^{H-\frac{1}{2}}\, dW_s\right),
\]
with the same values of $H$ and $\rho$. Although this smooth transformation does not remove the theoretical roughness constraint discussed in Remark~\ref{remark_Hroughness}, it improves the numerical performance and leads to more regular signature behavior. The total computation time (signature evaluation and parameter optimization) decreases from approximately $3$ hours to $39$ minutes, while achieving a comparable minimum value of the loss function (around $9 \times 10^{-4}$).

As a third alternative, we consider a shifted exponential transformation 
\[
X_t = X_0 \exp\left(\sqrt{2H}\int_0^t (t-s)^{H-\frac{1}{2}}\, dW_s\right),
\]
with an arbitrary initial value $X_0=0.1$ and the same values of $H$ and $\rho$. This specification offers a further improvement in computational efficiency, with the full calibration procedure requiring only $17$-$19$ minutes, and achieving a lower minimum value of the loss function ($3.5 \times 10^{-4}$), which indicates a better fit to market prices than the correlated Heston case. The corresponding optimal parameter vector is:
\begin{align*}
\ell^* = (& 0.17273586,\; -0.29578964,\; -0.08071348,\; 0.40101573,\; -0.2974647,\; 0.31988953,\; 1.40158411,\; \\
          & 0.15016936,\; -0.05769989,\; 0.00999173,\; 0.25021442,\; 0.02998332,\; -0.00789562,\; 0.12012242,\; \\
          & 0.27102252).
\end{align*}

\noindent From the simulated option prices we compute the implied volatilities, which we now denote by $\text{IV}^{\textsf{VIX}}$. Figure~\ref{figComparisonIVols3} shows the two surfaces.

\begin{figure}[H]
    \centering
    \includegraphics[width=0.5\linewidth]{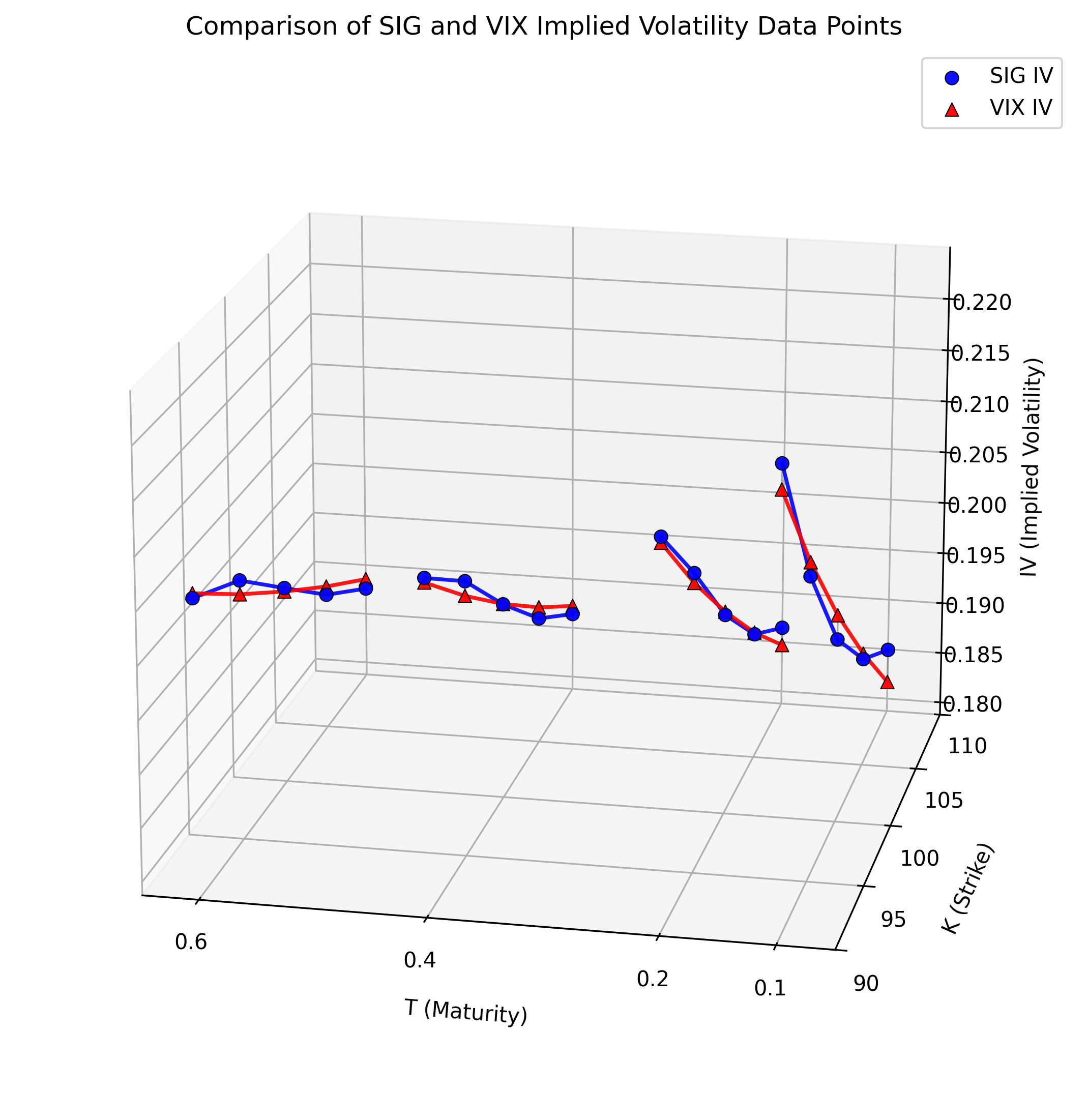}
    \caption{Implied volatility surfaces for the rough Bergomi model.}
    \label{figComparisonIVols3}
\end{figure}

For all 20 option contracts $(K_i, T_i)$, we compute the errors
\begin{align*}
    e^{\textsf{SIG}}_{i} \, & =\, \left| \text{IV}^{\textsf{SIG}}(K_i, T_i) - \text{IV}^{\textsf{mkt}}(K_i, T_i) \right| \\
    e^{\textsf{VIX}}_{i} \, & =\, \left| \text{IV}^{\text{VIX}}(K_i, T_i) - \text{IV}^{\textsf{mkt}}(K_i, T_i) \right|.
\end{align*}
A breakdown of the errors is provided in Table~\ref{tab:rBergomi_errors}, where entries marked with $(*)$ correspond to the cases in which $e^{\textsf{SIG}} < e^{\textsf{VIX}}$.

\begin{table}[!htbp]
\centering
\scriptsize
\begin{tabular}{ccll}
\toprule
$T$ & $K$ & $e^{\textsf{VIX}}$ & $e^{\textsf{SIG}}$ \\
\midrule
0.1 &  90 & 0.00104 & 0.00126 \\
0.1 &  95 & 0.00054 & 0.00183 \\
0.1 & 100 & 0.00001 & 0.00231 \\
0.1 & 105 & 0.00058 & 0.00004 (*) \\
0.1 & 110 & 0.00117 & 0.00442 \\
0.2 &  90 & 0.00079 & 0.00026 (*) \\
0.2 &  95 & 0.00041 & 0.00049 \\
0.2 & 100 & 0.00001 & 0.00033 \\
0.2 & 105 & 0.00044 & 0.00027 (*) \\
0.2 & 110 & 0.00087 & 0.00262 \\
0.4 &  90 & 0.00060 & 0.00019 (*) \\
0.4 &  95 & 0.00029 & 0.00107 \\
0.4 & 100 & 0.00003 & 0.00001 (*) \\
0.4 & 105 & 0.00035 & 0.00075 \\
0.4 & 110 & 0.00067 & 0.00015 (*) \\
0.6 &  90 & 0.00049 & 0.00094 \\
0.6 &  95 & 0.00022 & 0.00109 \\
0.6 & 100 & 0.00006 & 0.00039 \\
0.6 & 105 & 0.00032 & 0.00050 \\
0.6 & 110 & 0.00061 & 0.00037 (*) \\
\bottomrule
\end{tabular}
\caption{Calibration errors for the rough Bergomi case.} 
\label{tab:rBergomi_errors}
\end{table}

Both methods achieve a high level of calibration accuracy, with errors typically of order $10^{-4}$, indicating that the two approaches are broadly comparable in terms of precision. Note how the signature-based method outperforms the analytical approximation in a number of instances, as highlighted in the table, suggesting a slight advantage in this setting.

As in Section~\ref{sub:Heston_correlated}, we also increased the truncation level of the signature to $N=4$, but observed only marginal improvements, despite the additional computational cost. This indicates that the chosen truncation level already captures most of the relevant structure of the volatility surface.

A possible explanation for the improved performance of the signature-based approach with respect to the Heston case is the non-Markovian nature of the rough Bergomi model, driven by fractional Brownian motion. While analytical approximations rely on specific structural assumptions, the signature framework is naturally designed to encode temporal interactions, which may allow it to better represent such effects. Overall, these results further illustrate the robustness and adaptability of the signature-based approach in complex, non-Markovian volatility regimes.

\section{Conclusions}

This paper provides a detailed comparison between two complementary approaches to the calibration of implied volatility surfaces: analytical approximations and data-driven models based on signatures of rough paths. Rather than viewing these methodologies in opposition, our analysis highlights how they address the calibration problem from different but compatible perspectives, each with its own strengths in terms of structure, flexibility, and computational cost.

The analytical approach builds on model-specific frameworks (namely, the Heston and rough Bergomi models) and derives explicit calibration formulas: asymptotic expansions for Heston, and a new VIX-based calibration scheme for rough Bergomi introduced in this paper. When the underlying dynamics are known, these methods provide highly accurate calibration in a low-dimensional setting at minimal computational cost.

The signature-based methodology does not rely on a fixed parametric specification. Volatility is modeled as a linear functional of the signature of a primary process, which can be chosen to reflect different features of the data. This flexibility allows the method to adapt to a wide range of dynamics, including non-Markovian settings. In the Heston case, the signature-based model achieves a calibration accuracy comparable to the analytical approach, with the global optimization error over the whole implied volatility surface typically below $10^{-3}$.

When using a fractional Brownian motion as the primary process in a rough Bergomi setting, the calibration remains highly accurate, with global implied volatility errors consistently of order $10^{-4}$. The performance is slightly improved compared to the Heston-based specification, which may be attributed to the non-Markovian nature of fractional Brownian motion and the ability of signatures to capture such temporal dependencies effectively. These results further illustrate the robustness and adaptability of the signature-based approach in complex volatility regimes.

From a computational perspective, the analytical approximations are essentially instantaneous once derived, making them particularly attractive when the underlying model is specified. The signature-based approach, while more computationally demanding, remains practical. With $100{,}000$ simulated paths, the full calibration (including signature computation and optimization) takes approximately 15 minutes. Increasing the number of paths to $800{,}000$ improves accuracy at the cost of longer runtimes, between 45 and 90 minutes, depending on the model. This reflects a natural trade-off between precision and computational effort, and the method can be tuned according to the desired level of accuracy.

In summary, analytical methods provide an optimal solution when the model is correctly specified, combining precision and efficiency. Signature-based models, on the other hand, offer a robust and flexible alternative that performs well across different dynamics, particularly in non-Markovian settings. Together, these approaches balance model-driven insights with data-driven adaptability, opening promising directions for future research.

\vspace{0.8cm}

\bibliographystyle{apalike}
\bibliography{references.bib}

\end{document}